\documentclass[11pt,journal,onecolumn]{IEEEtran}
\usepackage{graphicx}
\usepackage{amsmath,amssymb,amsfonts,bm,algorithm,color,fmtcount,algorithmic,epsfig,subfigure,setspace,epstopdf}
\usepackage[finalnew]{trackchanges}
\usepackage[center]{caption}

\long\def\comment#1{}

\newfont{\bbb}{msbm10 scaled 700}

\newfont{\bb}{msbm10 scaled 1100}

\newcommand{\RR}{\mbox{\bb R}}

\newcommand{\fv}{{\bf f}}
\newcommand{\gv}{{\bf g}}

\newcommand{\uv}{{\bf u}}
\newcommand{\wv}{{\bf w}}

\newcommand{\zerov}{{\bf 0}}
\newcommand{\onev}{{\bf 1}}

\newcommand{\Am}{{\bf A}}

\newcommand{\Dm}{{\bf D}}

\newcommand{\Gm}{{\bf G}}
\newcommand{\Hm}{{\bf H}}
\newcommand{\Id}{{\bf I}}
\newcommand{\Jm}{{\bf J}}

\newcommand{\Lm}{{\bf L}}

\newcommand{\Pm}{{\bf P}}

\newcommand{\Tm}{{\bf T}}
\newcommand{\Um}{{\bf U}}

\newcommand{\Lam}{\bm{\Lambda}}

\newcommand{\Bc}{{\cal B}}

\newcommand{\Gc}{{\cal G}}

\newcommand{\Lc}{{\cal L}}

\newcommand{\Nc}{{\cal N}}
\newcommand{\Oc}{{\cal O}}

\newcommand{\Vc}{{\cal V}}

\newcommand{\Lcb}{\bm{\mathcal{L}}}

\addeditor{SN}
\title{Compact Support Biorthogonal  Wavelet Filterbanks for Arbitrary Undirected Graphs}
\author{Sunil~K.~Narang*,~\IEEEmembership{Student~Member,~IEEE,}
        and~Antonio~Ortega,~\IEEEmembership{Fellow,~IEEE}
\thanks{Sunil K. Narang is with the Signal \& Image Processing Institute, Ming Hsieh Department of Electrical
Engineering, University of Southern California, Los Angeles, California, USA 90089 email: narang.sunil@gmail.com.}
\thanks{Antonio Ortega is with the Signal \& Image Processing Institute, Ming Hsieh Department of Electrical
Engineering, University of Southern California, Los Angeles, California, USA 90089 email: antonio.ortega@sipi.usc.edu.}
\thanks{This work was supported in part by NSF under grant CCF-1018977.}}
\newtheorem{lemma}{Lemma}

\newtheorem{theorem}{Theorem}[section]
\newtheorem{proposition}{Proposition}

\definecolor{lightgray}{gray}{0.8}%
\newtheorem{corollary}[theorem]{Corollary}
\newenvironment{proof}[1][Proof]{\begin{trivlist}

\item[\hskip \labelsep {\bfseries #1}]}{\end{trivlist}}

\begin{document}

\maketitle
\begin{abstract}
In our recent work, we proposed
the design of 
perfect reconstruction orthogonal wavelet filterbanks, called 
{\em graph-QMF},  
for arbitrary undirected weighted graphs. In that
formulation we first designed ``one-dimensional'' two-channel filterbanks on 
bipartite graphs, and then extended them to ``multi-dimensional'' 
separable two-channel filterbanks for 
arbitrary graphs via a bipartite subgraph decomposition. 
We specifically designed wavelet filters based on 
the spectral decomposition of the graph, and stated 
necessary and sufficient 
conditions for a two-channel graph filter-bank 
on bipartite graphs 
to provide aliasing-cancellation, 
perfect reconstruction and orthogonal set of basis (orthogonality). 
While, the exact {\em graph-QMF} designs satisfy all the above conditions, 
they are not exactly $k$-hop localized on the graph.
In this paper, we relax the condition of orthogonality  
to design  a biorthogonal pair of graph-wavelets that can have
compact spatial spread and still satisfy the 
perfect reconstruction conditions. The design is analogous to the 
standard Cohen-Daubechies-Feauveau's (CDF) 
construction of factorizing a maximally-flat 
Daubechies half-band filter. Preliminary 
results demonstrate that the proposed filterbanks can be useful for both standard 
signal processing applications as well as 
for signals defined on arbitrary graphs. \\
{\bf Note:} Code examples from this paper are available at 
http://biron.usc.edu/wiki/index.php/Graph\_Filterbanks
\end{abstract}
\begin{center} \bfseries EDICS Category: DSP-WAVL, DSP-BANK, DSP-MULT, DSP-APPL, MLT \end{center}
\IEEEpeerreviewmaketitle
\section{Introduction}
\label{sec:intro}
\subsection{Motivation}
Graphs provide a flexible model for representing data in many domains
such as networks, computer vision, and high dimensional data-clouds.
The data on these graphs can be visualized as a finite collection of
samples, leading to a {\em graph-signal}, which can be defined as the
information attached to each node (scalar or vector values mapped to
the set of vertices/links) of the graph.  Examples include measured
values by sensor network nodes or traffic measurements on the links of
an Internet graph.  The formulation of datasets as graph-signals has
been subject to a lot of study recently, especially for the purpose of
analysis, compression and storage.  
Major challenges are posed by the size of these datasets, making them 
difficult to visualize, process or analyze. This has led to a recent 
interest in extending wavelet techniques to signals defined on
graphs. These techniques can provide multiresolution representations, 
so that smaller graphs with smooth approximations of the original
graph-signals can be obtained and used for processing. 
Moreover, these signal representations can be used 
to develop local
analysis tools, so that 
a graph-signal can be processed ``locally'' around a node or vertex, using data from a small neighborhood
of nodes around the given node. 
The trade-off between spatial and frequency localization
is fundamental in the study of wavelet transforms for regular signals.
This trade-off is being studied for graph signals as 
well~\cite{Hammond'09,SunilTSP,
  agaskar_icassp}. We can say that a graph signal, for example an elementary basis
function in a wavelet representation, is localized in the  {\em vertex domain} if most of the signal energy is concentrated in a 
{\em $k$-hop localized} neighborhood around a vertex, so that the
degree of localization will depend on how large $k$ is. 
Similarly, a graph signal can be said to be localized in the {\em
  spectral domain}, defined in terms of the eigenvalues and
eigenvectors of the {\em graph Laplacian matrix}, if the projection of
the signal onto these eigenvectors has most of its energy
concentrated in a band of graph-frequencies around a center frequency.

There has been a significant recent interest in the extension of
wavelet transforms to graph signals, including wavelets
on unweighted graphs for analyzing computer network traffic
\cite{Crovella'03}, diffusion wavelets and diffusion wavelet packets
\cite{Coifman'06, Maggioni_biorthogonal,bremer_packets}, the
``top-down'' wavelet construction of \cite{szlam}, graph dependent
basis functions for sensor network graphs~\cite{Ramchandran'06},
lifting based wavelets on graphs \cite{GodwinJ,Silverman,
  narang_lifting_graphs}, multiscale wavelets on balanced trees
\cite{gavish}, spectral graph wavelets \cite{Hammond'09}, and our
recent work on graph wavelet filterbanks \cite{SunilTSP}. 

In this paper, we focus on the problem of designing wavelet transforms
that are {\em invertible, compactly supported on the graph and
 critically sampled (CS).} Critical sampling, i.e., the number of
wavelet coefficients generated by the transform is equal to the number
of vertices in the graph, is important in order to achieve a compact
representation (e.g., for compression). Also, in a critically sampled
transform a subset of vertices will include low frequency
information. This leads to a natural approximation of the original
graph by a smaller graph (containing only those vertices and their
corresponding graph signal). 

Of the above mentioned approaches,
only the lifting based approaches~\cite{GodwinJ,Silverman}, the
wavelets on balanced trees \cite{gavish} and the graph wavelet
filterbanks \cite{SunilTSP} achieve critical sampling, although in all
cases there are some restrictions on the kinds of graphs on which the
transforms can be defined.

In terms of localization, the lifting-based
approaches~\cite{GodwinJ,Silverman} are compactly supported in the
graph, i.e., their basis functions are {\em strictly} $k$-hop
localized the vertex domain, which implies that the output at each
vertex $n$ can be computed exactly from the data at the vertex $n$ and
a $k$-hop neighborhood around it.  Instead, in our filterbank-based
approach \cite{SunilTSP} basis functions had good vertex domain
localization but were not compactly supported. Compact support could
be achieved by approximating their spectral response by a polynomial,
but this comes at the expense of introducing a small reconstruction
error (i.e., there was no longer perfect reconstruction). 

The main contribution in this paper is to extend the filterbank
approach to design graph wavelets of \cite{SunilTSP} so as to achieve
compact support.
Note that the lifting based
techniques~\cite{GodwinJ,Silverman} are also compactly supported but
they are vertex domain designs for which it is difficult to control
the quality (e.g., localization) of their spectral
representation. Instead, in the work presented here filters are
designed 
in the spectral domain while guaranteeing compact vertex domain support, 
so that we can directly 
control explicitly the trade-off between localization in the vertex domain and
the spectral domain.



As in \cite{SunilTSP}, the building blocks our design are 
{\em two channel wavelet filterbanks on  bipartite graphs}, which
provide a decomposition of any graph-signal 
into a lowpass (smooth) graph-signal, 
and a highpass (detail) graph-signal. 
For arbitrary graphs the filterbanks can be extended in two ways: a)
by implementing the proposed wavelet filterbanks on a 
bipartite graph approximation of the original graph, which provides a ``one-dimensional'' analysis, or b)
by decomposing the graph into multiple link-disjoint bipartite
subgraphs, and applying the proposed filterbanks 
iteratively on each the subgraphs (or on some of them), leading to  a
``multi-dimensional'' analysis \cite{SunilTSP}.

In \cite{SunilTSP} we showed that downsampling/upsampling operations
in bipartite graphs lead to a {\em spectral folding} phenomenon, which
is analogous to aliasing in regular signal domain.  We utilized this
property to propose two channel critically sampled wavelet
filterbanks, called {\em graph-QMF}, on arbitrary undirected weighted
graphs.  We specifically designed wavelet filters based on the
spectral decomposition of the graph, and stated necessary and
sufficient conditions for a two-channel graph filterbank on bipartite
graphs to provide aliasing-cancellation, perfect reconstruction and
orthogonal set of basis (orthogonality).  While the exact {\em
  graph-QMF} designs satisfy all the above conditions, they are not
compactly supported 
on the graph. 
%
%
In order to design 
compactly supported 
graph-QMF 
transforms, we performed a Chebychev polynomial approximation of the 
exact filters in the spectral domain, and this incurred error in 
the reconstruction of the signal, and loss of orthogonality.
Here, we propose an alternative to graph-QMF design
where we relax the conditions of orthogonality, and 
design  a biorthogonal pair of graph-wavelets, which we call {\em
  graphBior}.  These new designs lead to a representation  
that is exactly $k$-hop localized, 
and still satisfies the 
perfect reconstruction conditions. This design is analogous to 
the standard Cohen-Daubechies-Feauveau's (CDF) \cite{CDF9}
construction to obtain maximally half-band filters. Even though these filterbanks 
are not orthogonal, 
we show that they can be designed to nearly preserve energy.
In particular, we compute expressions for Riesz bounds of the filterbanks, 
and choose graph-wavelets with the maximum ratio of lower and upper Riesz bounds.

We will show that it is possible to design these filterbanks based on
both the normalized Laplacian and the random-walk Laplacian, leading
to {\em nonzeroDC graphBior } and {\em zeroDC graphBior} designs, respectively. As 
the name suggests, the
{\em zeroDC} design has the advantage of leading to highpass operators with
a zero response for the {\em all constant} signal, which will be
useful in applications in Euclidean space (e.g., when applying graph wavelets to
regular domain signals such as images). Instead, in the {\em nonzeroDC}
design the DC frequency corresponds to a signal where the value
at each node depends on its degree. Finally, because our designs are
biorthogonal, the norms of highpass and lowpass are not equal. For
applications where a normalization is required, 
we propose 
 unity gain compensation (GC) techniques for both types of designs. 
A comparison of our proposed design vis-a-vis existing transforms is shown in Table~\ref{tab:conclusion_tab}.

\addtocounter{footnote}{1}
\footnotetext[\value{footnote}]{The exact Graph-QMF solutions are perfect reconstruction 
and orthogonal, but they are not compact support. 
Localization is achieved with a matrix polynomial approximation 
of the original filters, which incur some
loss of orthogonality and reconstruction error, 
which can be arbitrarily reduced by increasing the degree of approximation.}
\addtocounter{footnote}{-1}
\begin{table}[htb]
{\small
\centering
\begin{tabular}{|p{6 cm}|p{0.5 cm}|p{0.5 cm}|p{0.5 cm}|p{0.9 cm}|p{0.5 cm}|p{0.5 cm}|}
\hline
Method & DC & CS & PR & Comp & OE & GS \\ 
\hline
Wang \& Ramchandran~\cite{Ramchandran'06} & No & No & Yes & Yes & No & No \\ 
\hline
Crovella \& Kolaczyk~\cite{Crovella'03} & Yes & No & No & Yes &  No & No \\ 
\hline
Lifting Scheme~\cite{GodwinJ} & Yes & Yes & Yes & Yes & No & Yes \\ 
\hline
Wavelets on balanced trees~\cite{gavish} & Yes & Yes & Yes & Yes & No & Yes \\ 
\hline
Diffusion Wavelets~\cite{Coifman'06} & Yes & No & Yes & Yes & Yes & No \\ 
\hline
Spectral Wavelets~\cite{Hammond'09} & Yes & No & Yes & Yes & No & No \\ 
\hline
{graph-QMF filterbanks (exact)~\cite{SunilTSP}} & No & Yes  & Yes & No & Yes & No \\ 
\hline
{graph-QMF filterbanks (approx.)~\cite{SunilTSP}} & No & Yes  & No\addtocounter{footnote}{+1}$^{\decimal{footnote}}$ & Yes & No$^{\decimal{footnote}}$ & No \\ 
\hline
{\bf {\em nonzeroDC graphBior filterbanks}} & {\bf No} & {\bf Yes}  & {\bf Yes} & {\bf Yes} & {\bf No} & {\bf No} \\ 
\hline
{\bf {\em zeroDC graphBior filterbanks}} & {\bf Yes} & {\bf Yes}  & {\bf Yes} & {\bf Yes} & {\bf No} & {\bf No} \\ 
\hline
\end{tabular}
\caption{Comparison of graph wavelet designs in terms of key
  properties: zero highpass response for constant graph-signal (DC), 
critical sampling (CS), perfect reconstruction (PR), 
 compact support (Comp), orthogonal expansion (OE), requires graph
 simplification (GS).}
\label{tab:conclusion_tab}
}
\end{table}%

The outline of the rest of the paper is as follows: in Section \ref{sec:prelim}, we introduce some notations, 
a general formulation of two channel wavelet filterbank
on graphs and the graph-QMF filterbanks proposed in \cite{SunilTSP}, which are orthogonal 
and perfect reconstruction. In Section \ref{sec:nonzeroDC_gb}, we describe proposed 
nonzeroDC graphBior filterbanks on graphs, and in Section \ref{sec:zerodc_gb} we design and describe the properties
of zeroDC graphBior filterbanks. In Section~\ref{sec:multi_dimensional}, we describe extension of proposed filterbanks 
to arbitrary graphs via bipartite subgraph decomposition, and multiresolution implementation.
%
In Section~\ref{sec:experiments}, 
we conduct some experiments to demonstrate the properties and applications 
of the proposed filterbanks. Section~\ref{sec:conclusion} concludes
our paper.

\section{Preliminaries}
\label{sec:prelim}
A graph can be denoted as $G = (\Vc,\textit E)$  
with vertices (or nodes) in set $\Vc$ and 
links  as tuples $(i,j)$ in $\textit E$. 
The graphs considered in this work are 
undirected graphs without self-loops and without multiple links between nodes.
The links can only have positive weights.
The size of the graph $N = |\Vc|$ 
is the number of nodes and  \add{the}
geodesic 
distance metric is given as $d_{\Gc}(i,j)$\add{, which represents 
the sum of link weights along the shortest path between nodes $i$ and $j$,
and is considered infinite if $i$ and $j$ are disconnected. } 
Define $\Nc_{j,n}$ as the $j$-hop neighborhood of node $n$, i.e., 
$\Nc_{j,n} = \{m \in \Vc,d_{\Gc}(i,j) \leq j\}$.
Denote the identity matrix as $\Id$, and let $\delta_n$ be an impulse function, i.e., $\delta_n(n) = 1$ and $\delta_n(m) = 0$
for all $m \neq n$.  Define  $<\fv,~\gv> = \fv^\top\gv$ as the inner-product of vector $\fv$ and $\gv$, where 
$(.)^\top$ is the transpose operator.
Define
$\Am = [w_{ij}]$, 
the adjacency matrix of the graph, $d_i$ the degree (sum of link-weights) of node $i$, 
and $\Dm = diag\{{d_i}\}_{i = 1,2,...N}$, the diagonal degree matrix of graph, so 
that $\Lm = \Dm - \Am$ is the {\em unnormalized Laplacian matrix} of the graph.

\subsection{Graph vertex domain}
We define a 
signal $f: \Vc \to \RR$ on a graph as a set of scalars, where each 
scalar is assigned to one of the vertices of the graph. This can be extended to 
vector values at each vertex. 
Further, a graph based transform is defined as
a linear transform $\Tm:\RR^N \to \RR^M$ 
applied in the vertex domain,
such that the operation at each 
node $n$ is a linear combination of the
value of the graph-signal $f(n)$ at the node $n$ and the values $f(m)$ on 
nearby nodes $m \in \Nc_{j,n}$,  
i.e.,
\begin{equation}
y(n) = <\Tm(n,.),~~\fv> = T(n,n) f(n) + \sum_{m \in \Nc_{j,n}} T(n,m) f(m),
 \label{eq:local_linear_tx}
\end{equation}
where $\Tm(n,.)$ is the $n^{th}$ row of the transform $\Tm$.
In analogy to the $1$-D regular case, we would sometimes 
refer to graph-transforms as graph-filters, and $\Tm(n,.)$ for $n = 1,2,...N$ as the {\em impulse response}
of the transform $\Tm$ at the $n^{th}$ node. Note that due to the irregularity of links,  the impulse response 
of a graph filter varies from one vertex to the other.
A desirable feature of graph filters is {\em spatial localization}, which typically means that 
the energy of each impulse response (i.e., each row) of the graph filter is 
concentrated in a local region around a node. 
In this paper, we use the definition proposed in~\cite{agaskar_icassp} to define the spatial spread of any 
signal $\fv$ on a graph $G$ as:
\begin{equation}
 \Delta_{G}^2(\fv) := \inf_i \frac{1}{||\fv||_2^2}\sum_{j \in \Vc}[d_{\Gc}(i,j)]^2[f(j)]^2.
\label{eq:spatial_spread}
\end{equation}
Here, $\{[f(j)]^2/||\fv||^2\}_{j=1,2,\ldots,N}$ can be interpreted 
as a probability mass function (pmf) of signal $\fv$,  
and  $\Delta_{G}^2(\fv)$ 
is the variance of the geodesic distance function $d_{\Gc}(i,.):\Vc \to \RR$  at node $i$, in terms of this spatial pmf.
Thus, $\Delta_{G}^2(\Tm(n,.))$ should be small for all $n = 1,2,...N$ for good spatial localization. 
In our analysis, 
we compute the 
spatial spread of a graph transform $\Tm$ to be the average of the spatial spread (\ref{eq:spatial_spread}) of impulse responses 
over all vertices, i.e., 
\begin{equation}
 \Delta_{G}^2(\Tm) := \frac{1}{N}\sum_{n=1}^N\Delta_{G}^2(\Tm(n,.))
\label{eq:spatial_spread_tx}
\end{equation}

\subsection{Graph spectral domain}
As in our previous design~\cite{SunilTSP}, 
we use the symmetric normalized Laplacian 
matrix $\Lcb =\Dm^{-1/2}\Lm\Dm^{-1/2}$ 
to define spectral properties of the graph.
%
%
%
In this paper, we use the same Laplacian matrix $\Lcb$ 
to design nonzeroDC filterbanks.
Because 
$\Lcb$ is a real symmetric matrix,
it has a complete set of orthonormal eigenvectors, which we
denote by $\{\uv_l\}_{l = 0,1,2,...,N-1}$. These eigenvectors have associated
real, non-negative eigenvalues $\{\lambda_l\}_{l = 0,1,2,...,N-1}$ satisfying
$\Lcb\uv_{l} =\lambda_{l}\uv_{l}$, for $l = 0,1,2,...,N-1 $, and ordered as:
$\lambda_0 \leq \lambda_1 \leq ...\lambda_{N-1}$.
In our analysis, we assume the graph to be connected 
\footnote{For graphs with multiple separate connected components, we analyze each component as a separate graph.}.
Zero appears as a unique minimum
eigenvalue 
of the graph and the maximum eigenvalue is always less than or equal to $2$, with equality 
if and only if the graph is a bipartite graph {\cite[Section 2]{eigenvaluespacings}}. 

We denote 
the {\em spectrum of the graph} by $\sigma(\Lcb) := \{\lambda_0,\lambda_1,...\lambda_{N-1}\}$.
A graph signal $\fv$  is represented in the spectral domain as its projection onto the eigenvectors, 
denoted as : $ \hat f(\lambda_l) = <\fv,~\uv_l>$.
An eigenspace $V_{\lambda}$ is defined as the space spanned by eigenvectors of $\Lcb$ 
associated with eigenvalue $\lambda$, and the {\em eigenspace projection matrix} is defined as:
\begin{equation}
 \Pm_{\lambda} := \sum_{\lambda_l = \lambda} \uv_{l}\uv_l^\top, \nonumber
\end{equation}
where $\uv_l^\top$ is the transpose of eigenvector $\uv_{l}$. 
The eigenspace projection matrices 
are idempotent and $\Pm_{\lambda}$ and $\Pm_{\gamma}$ 
are orthogonal if $\lambda$ 
and $\gamma$ are distinct eigenvalues of the Laplacian matrix, i.e., 
\begin{equation}
\Pm_{\lambda}\Pm_{\gamma} = \delta(\lambda - \gamma)\Pm_{\lambda},
\label{eq:eigenspace_prop1}
\end{equation}
where $\delta(\lambda)$ is the Kronecker delta function. Further the sum of all eigenspace projection matrices 
for any graphs is an identity matrix.  
Closely
related to the symmetric Laplacian matrix $\Lc$ is the random-walk graph Laplacian,
which is defined as $\Lcb_r := \Dm^{-1}\Lm$. We 
use $\Lcb_r$ to design zeroDC filterbanks.
Note that $\Lcb_r$ has the same set of 
eigenvalues as $\Lcb$, and if 
$\uv_l$ is an eigenvector 
of $\Lcb$ associated with $\lambda_l$, 
then $\Dm^{-1}\uv_{l}$ 
is an eigenvector of $\Lcb_r$ 
associated with the eigenvalue $\lambda_l$.
Similar to~(\ref{eq:spatial_spread}), 
the spectral spread of a graph signal c
an be defined as:
\begin{equation}
 \Delta_{\sigma}^2(\fv) := \min_{\mu {\in \mathbb{R}_+}} \left\{\frac{1}{||\fv||_2^2}\sum_{\lambda \in \sigma(\Lcb)}\left[\lambda - \mu\right]^2\left[\hat f(\lambda)\right]^2\right\},
\label{eq:spectral_spread}
\end{equation}
where $\{[\hat f( \lambda)]^2/||\fv||_2^2\}_{\lambda=\lambda_0,\lambda_1,\ldots,\lambda_{N-1}}$  is the pmf of $\fv$ 
across the spectrum of the Laplacian matrix, and $\mu$ is the mean of $\lambda$ with respect to 
this pmf. $\Delta_{\sigma}^2(\fv)$ computes the
variance of $\lambda$ with respect to the spectral pmf function. 
Thus, for a signal to have good spectral localization, the value of 
 $\Delta_{\sigma}^2(\Tm(n,.)) $ should be small for all vertices.
 In our analysis, we compute 
 the spectral response of any transform 
 $\Tm$ as the average spectral response 
 of impulse responses at all vertices, i.e., 
\begin{equation}
 |\hat \Tm(\lambda)|^2 := \frac{1}{N}\sum_{i=1}^N |\hat \Tm(n,.)(\lambda)|^2,
\end{equation}
and use $|\hat \Tm(\lambda)|^2$ as the spectral pmf function
to compute the spectral spread of $\Tm$\footnote{
Note that the definitions of spread 
presented here are heuristic and do not have a 
well-understood theoretical background. Another definition of spectral spread in graphs 
is given in~\cite{agaskar_icassp}.
If the graph is not regular, the choice of which Laplacian matrix ($\Lcb$ or $\tilde \Lcb$) 
to use for computing spectral spreads also affects the results. The purpose of these definitions and the subsequent examples is to show 
that 
a trade-off exists between spatial and spectral localization in graph wavelets.}.

\subsection{Spectral graph filters}
For designing compact support wavelet filters 
on graphs we use the same approach as in~\cite{SunilTSP},
and define 
analysis wavelet filters $\Hm_i$ and $\Gm_i$ 
in terms of spectral kernels $\hat h_i(\lambda)$ and $\hat g_i(\lambda)$ for $i = 0,1$ respectively. 
The corresponding 
transform matrices are represented as:
\begin{equation}
 \begin{array}{l}
 \displaystyle {\Hm}_i = \hat h_i({{\bf \Lc}}) = \sum_{\lambda \in \sigma(G)} \hat h_i(\lambda){\Pm}_{\lambda},   \\
 \displaystyle  {\Gm}_i = \hat g_i({\bf \Lc}) = \sum_{\lambda \in \sigma(G)}\hat g_i(\lambda){\Pm}_{\lambda}.
 \end{array}
\label{eq:spectral_tx}
\end{equation} 
These filters have the following interpretation: 
the output of a spectral filter with kernel $\hat h(\lambda)$ can be expanded as
$\fv_H = \Hm \fv =  \sum_{\lambda \in \sigma(G)} \hat h_i(\lambda)~\Pm_{\lambda}\fv$, 
where $\fv_{\lambda} = \Pm_{\lambda}\fv$ is the component of 
input signal $\fv$ in the $\lambda$-eigenspace. Thus, filter $\Hm$ 
either attenuates or enhances different harmonic 
components of the input signal depending upon the 
magnitude of $\hat h(\lambda)$. 
Therefore, we will also refer to $\hat h(\lambda)$ as the {\em spectral response} of filter $\Hm$. 
%
For a general kernel function, the 
filtering operations corresponding to $\Hm_i$ and $\Gm_i$ 
may not have compact support, and would
require a full spectral 
decomposition of Laplacian matrix. However it has been shown in \cite{Hammond'09}, 
that the spectral response can be approximated as a
polynomial of degree $K$, and 
the corresponding filters can be computed iteratively with $K$ one-hop operations at each node. 
Further, any graph filters with 
a $K$ degree polynomial spectral response are exactly
{\em $K$-hop localized} (compact support)~\cite[Lemma 5.2]{Hammond'09},
and can be efficiently computed without diagonalizing the Laplacian matrix. 
The computational complexity of the filtering operations in the polynomial case, reduces to 
$\Oc(K|E|)$ for degree $K$ and  $|E|$ number of links in the graph. 
Thus, the degree $K$ in case of
polynomial spectral response can be interpreted as the length of the corresponding spectral filters
\footnote{Note that, having a polynomial spectral response for compact support is necessary
only in case of spectral graph filters. There can be non-spectral graph-filters, for example, 
graph wavelets proposed in~\cite{Crovella'03}, that have compact support without being a polynomial in the spectral domain.}.
\subsection{Spectral wavelet filterbanks}
In \cite{SunilTSP}, we described the construction of a
two-channel wavelet filterbank on a bipartite graph $\Bc = (L,H,E)$\footnote{A bipartite graph $G = (L,H,\textit E)$ is a graph whose 
vertices can be divided into two disjoint 
sets $L$ and $H$, such that 
every link connects a vertex 
in $L$ to one in $H$.},  
characterized by filtering operations
$\{\Hm_i,\Gm_i\}_{i = 0,1}$ and a  function  $\beta(n)$, 
which provides a
decomposition of graph-signal $\fv$ into 
a lowpass (approximation) graph-signal $\fv_{low}$
and a highpass (details) graph-signal component $\fv_{high}$.  
\begin{figure}[htb]
\centering
\includegraphics[width=4in]{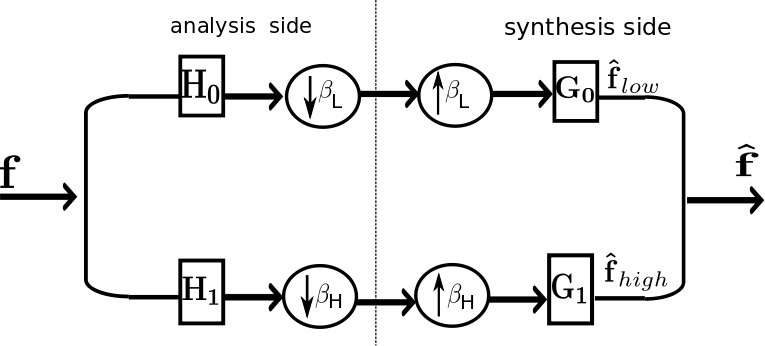}
\caption{\footnotesize Block diagram of a two-channel wavelet filterbank on graph. 
}
\label{fig:graph_filterbank}
\end{figure}
The transforms $\Hm_i,\Gm_i$ for $i=0,1$ of the two channels 
are {\em graph transforms} with spectral kernels $h_i$ and 
$g_i$ respectively, as given in~(\ref{eq:spectral_tx}). 
In the analysis side of the filterbank, the input signal 
is first operated upon by transform $\Hm_0$ in the lowpass channel and 
$\Hm_1$ in the high pass channel. A subsequent downsampling upsampling (DU)
operation, discards and replaces with zeros the output coefficients 
on the set $H$ in the lowpass channel and 
on the set $L$ in the highpass channel. Since $L$ and $H$ are disjoint and complementary 
subsets of vertex set $\Vc$, the retained set of output coefficients
is critically sampled. Algebraically, the DU operation can be represented 
with a function $\beta(n)$, such that $\beta(n) = 1$, if node $n \in L$ and 
$\beta(n) = -1$ if node $n \in H$. 
Thus, the DU operation in the lowpass channel is 
given as $\frac{1}{2}(1+\beta(n))$ and in the highpass channel $\frac{1}{2}(1-\beta(n))$.
It can also be written in the matrix form as $\frac{1}{2}(\Id+\Jm_{\beta})$
for lowpass channel and $\frac{1}{2}(\Id-\Jm_{\beta})$ for highpass channel, 
where $\Jm_{\beta} = diag\{\beta\}$ is a diagonal matrix. The output signal $\fv_{du}$
of the DU operation in the two channels is expressed in Figure~\ref{fig:DU_figure}.
\begin{figure}[htb]
\centering
\includegraphics[width=4in]{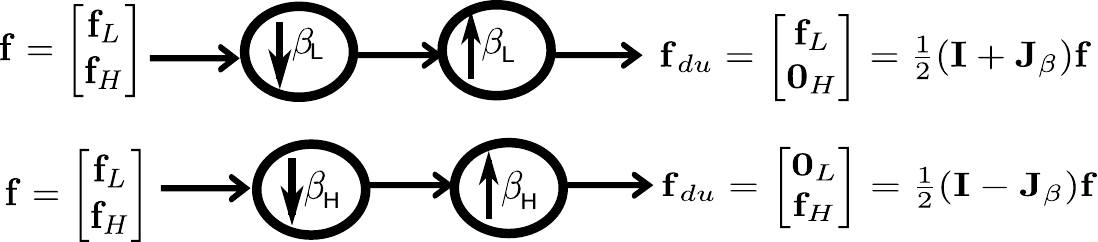}
\caption{\footnotesize DU operations in the proposed two-channel filterbank. 
}
\label{fig:DU_figure}
\end{figure}

%
%
%
%
%
The overall transfer matrix of the two channel filterbank 
can be written as:
\begin{eqnarray}
 \displaystyle \Tm & = &  \frac{1}{2}\Gm_0(\Id + \Jm_{\beta})\Hm_0 + \frac{1}{2}\Gm_1(\Id - \Jm_{\beta})\Hm_1  \nonumber \\
& = &  \underbrace{\frac{1}{2}(\Gm_0\Hm_0 + \Gm_1\Hm_1)}_{\Tm_{eq}} +  \underbrace{\frac{1}{2}(\Gm_0\Jm_{\beta}\Hm_0 - \Gm_1\Jm_{\beta}\Hm_1)}_{\Tm_{alias}}, 
\label{eq:overall_tx}
\end{eqnarray}
where $\Tm_{eq}$ is the transfer function of the filterbank without the $DU$ operation and 
$\Tm_{alias}$ 
arises primarily due to the $DU$ operations. 
Using the above formulation, we derived 
the following results\footnote{see \cite{SunilTSP} for proofs and details.}:

\begin{theorem}[{\bf Spectral folding phenomenon~\cite[Prop. $1$]{SunilTSP}}]
Given a bipartite graph $\Bc = (L,H,\textit E)$ with Laplacian matrix $\Lcb$
and with the binary function $\beta$ defined as above,
%
if $\uv_{\lambda}$ is the eigenvector of a unique eigenvalue $\lambda$ of graph $\Bc$ then 
\begin{equation}
\Jm_{\beta}\uv_{\lambda} = \pm \uv_{2 - \lambda}.
 \label{eq:spectral_folding_ev}
\end{equation}
The ambiguity of $\pm$ sign appears in~(\ref{eq:spectral_folding_ev}) since both 
$\uv_{\lambda}$ and $-\uv_{\lambda}$ can be eigenvectors of $\Bc$ for eigenvalue $\lambda$.
Further, if the eigenvalue $\lambda$ appears with multiplicity greater than $1$, and if
$\Pm_{\lambda}$
is the projection matrix corresponding to the eigenspace $V_{\lambda}$, 
then 
\begin{equation}
\Jm_{\beta}\Pm_{\lambda} = \Pm_{2 - \lambda}\Jm_{\beta}.
 \label{eq:spectral_folding}
\end{equation}
 \label{thm:spectral_folding}
\end{theorem}
According to~(\ref{eq:spectral_folding_ev}) and~(\ref{eq:spectral_folding})
the eigenvector (or eigenspace) of eigenvalue $\lambda$ in a bipartite graph
changes to the eigenvector (or eigenspace) of eigenvalue $2 - \lambda$ after 
multiplying with $\Jm_{\beta}$. 
In other words, the eigenspace folds across the 
imaginary axis at $\lambda = 1$, 
hence we call this a {\em spectral folding phenomenon}\footnote{Note that spectral folding 
phenomenon only occurs if the binary function $\beta(n)$ is defined on one of the natural partitions $L$ or $H$ of 
the bipartite graph $\Bc = (L,H,E)$, and not for any other partition.}.

\begin{theorem}[{\bf Perfect reconstruction property~\cite[Sec. III.B]{SunilTSP}}]
\label{thm:PR}
Given a bipartite graph $\Bc = (L,H,\textit E)$, the filtering operations $\{\Hm_i,\Gm_i\}$ and the binary
function $\beta(n)$ as defined above, a necessary and sufficient condition 
for the {\em perfect reconstruction} in the two channel filterbanks is 
that for all $\lambda$ in $\sigma(\Bc)$,
\begin{eqnarray}
\hat g_0(\lambda)\hat h_0(\lambda) + \hat g_1(\lambda) \hat h_1(\lambda) = 2, \nonumber \\
\hat g_0(\lambda)\hat h_0(2-\lambda) - \hat g_1(\lambda)\hat h_1(2-\lambda) = 0.
 \label{eq:perfect_reconstruct0}
\end{eqnarray}
 \end{theorem}
The advantage of representing PR conditions as in~(\ref{eq:perfect_reconstruct0}) is that the filterbank 
can be designed in the spectral domain of the graph by designing spectral kernels which satisfy~(\ref{eq:perfect_reconstruct0}). These kernels 
can be designed as continuous functions of $\lambda \in [0~ 2]$, which obviates the need to evaluate~(\ref{eq:perfect_reconstruct0}) 
only at the spectrum $\sigma(\Bc)$, thus making the design independent 
of the structure of the graph. 
The {\em graph-QMF} solution provided in~\cite{SunilTSP}, satisfies~(\ref{eq:perfect_reconstruct0}), and is also orthogonal. While there exist 
many exact graph-QMF solutions, we show in Section~\ref{sec:half-band kernel} (see last paragraph) that none of the solutions 
has compact support. 
Therefore, we shift our focus to biorthogonal solutions which are PR and have compact support.

\section{nonzeroDC graphBior filterbanks}
\label{sec:nonzeroDC_gb}
Given the parameter $k$, in order to design $k$-hop localized 
filters
to satisfy the perfect reconstruction conditions
%
given in (\ref{eq:perfect_reconstruct0}),
we need to design 
four polynomial spectral kernels of degree $k$, 
namely lowpass analysis kernel $\hat h_0(\lambda)$, highpass analysis kernel $\hat h_1(\lambda)$,
lowpass synthesis kernel $\hat g_0(\lambda)$, 
and highpass synthesis kernel $\hat g_1(\lambda)$. 
If we choose analysis and synthesis highpass kernels to be:
\begin{eqnarray}
\displaystyle \hat h_1(\lambda) & = & \hat g_0(2- \lambda) \nonumber \\
\displaystyle \hat g_1(\lambda) & = & \hat h_0(2- \lambda), 
 \label{eq:perfect_reconstruct2}
\end{eqnarray}
then, (\ref{eq:perfect_reconstruct0}) reduces to a single constraint for all eigenvalues, given as:
\begin{equation}
\hat h_0(\lambda)\hat g_0(\lambda) + \hat h_0(2 - \lambda)\hat g_0(2 -\lambda) = 2.
 \label{eq:biorthogonal_perfect_reconstruct}
\end{equation}
Further, define $\hat p(\lambda) =  \hat h_0(\lambda)\hat g_0(\lambda)$, 
then~(\ref{eq:biorthogonal_perfect_reconstruct}) can be written as:
\begin{eqnarray}
\hat p(\lambda) + \hat p(2 - \lambda) = 2.
 \label{eq:p_lambda}
\end{eqnarray}
In our approach, we first design $\hat h_0(\lambda)$ and $\hat g_0(\lambda)$, 
and then $\hat h_1(\lambda)$ and $\hat g_1(\lambda)$ can be obtained 
using~(\ref{eq:perfect_reconstruct2}). Further, since
$\hat p(\lambda)$ 
is the product of two lowpass kernels, it is also a lowpass kernel. Therefore, the objective is to design $\hat p(\lambda)$
as a polynomial {\em half-band kernel}\footnote{An ideal half band kernel $h(\lambda)$ in the 
spectral domain of a bipartite graph can be defined as a kernel with 
$h(\lambda) = 1$ for $\lambda \leq 1$, and 
$h(\lambda) = 0$ otherwise.}, which 
satisfies~(\ref{eq:biorthogonal_perfect_reconstruct}), and then obtain kernels $\hat h_0(\lambda)$ 
and $\hat h_1(\lambda)$ via spectral factorization.  
This  design is similar to 
the design proposed 
by Cohen-Daubechies-Feauveau~\cite{CDF9}
for finding a maximally 
flat pair of lowpass and highpass filters under the 
given length constraint, and 
then dividing up the residual factors between the two filters 
in a way that makes the basis function nearly orthogonal.
The following result is useful in our analysis:
\begin{proposition}
\label{prop:odd_p_lambda}
 If $\hat h_0(\lambda)$ and $\hat g_0(\lambda)$ are polynomial kernels, 
then any $\hat p(\lambda) = h_0(\lambda)g_0(\lambda)$, which satisfies (\ref{eq:p_lambda}) for all $\lambda \in [0~2]$, is an odd degree polynomial. 
\end{proposition}
\begin{proof}
By changing the variable so that $\lambda = 1 + l$, we can write (\ref{eq:p_lambda}) as:
\begin{equation}
 \hat p(1 + l)+ \hat p(1 - l) = 2,
  \label{eq:p_lambda_cv}
\end{equation}
where $\hat p(1 + l) = \hat h_0(1 + l) \hat g_0(1 + l)$ and $l \in [-1~1]$.  
If $\hat h_0(l)$ and $\hat g_0(l)$ are polynomial in $l$
then the functions 
$\hat p(1 + l)$ and $\hat p(1 - l)$ are also polynomials in $l$, and 
can be expressed as:
  \begin{eqnarray}
  \hat p(1 + l) &=& \sum_{k=0}^K c_k (l)^k, \nonumber \\ 
  \hat p(1 - l) &=&  \sum_{k=0}^K c_k (-l)^k. 
    \label{eq:prop1pf1}
  \end{eqnarray}
Using (\ref{eq:prop1pf1}) in (\ref{eq:p_lambda_cv}), we get:
\begin{equation}
\hat p(1 + l) + \hat p(1 - l) = \sum_{k=0}^K c_k ((l)^k + (-l)^k) =  2c_0 + \sum_{k=1}^{K/2} c_{2k}l^{2k}.
 \label{eq:prop1pf2}
\end{equation}
Thus $\hat p(1 + l) + \hat p(1 - l)$ is an even polynomial function of $l$. In order for 
both~(\ref{eq:p_lambda_cv}) and (\ref{eq:prop1pf2}) to be true 
%
for all $l \in [-1~1]$, 
$c_0 = 1$
and all other even power coefficients $c_{n}$ in the polynomial expansion of $\hat p(1 + l)$ must be $0$.
Therefore, the solution
$\hat p(1 + l)$, expressed as:
\begin{equation}
 \hat p(1 + l) = 1 + \sum_{n = 0}^{K} c_{2n+1}l^{2n+1}, 
 \label{eq:prop1pf5}
\end{equation}
is an odd degree polynomial. Thus, ignoring the trivial case $\hat p(1 + l) = 1$, the highest degree of 
$\hat p(1 + l)$ and $\hat p(1 - l)$  (and hence $\hat p(\lambda)$) is always odd. 
%
%
\end{proof}

\subsection{Designing half-band kernel $\hat p(\lambda)$}
\label{sec:half-band kernel}
The following known 
results help us prove the existence 
of a polynomial $\hat p(\lambda)$ that satisfies (\ref{eq:p_lambda_cv}), and obtain its spectral 
factorization:
\begin{lemma}[Bezout's identity {\cite[prop.~3.13]{Vetterli_book}}] 
Given any two polynomials $a(l)$ and $b(l)$ of continuous variable $l$,
\begin{equation}
a(l) x(l)  + b(l) y(l)  = c(l),
 \label{eq:bezout}
\end{equation}
has a solution $[x(l ),~y( l ) ]$,  if and only if  $gcd( a ( l ) ,  b ( l ) )$  divides $c ( l )$, where 
$gcd( a(l),b(l))$ refers to the greatest common divisor of polynomials $a(l)$ and $b(l)$.   
\label{lem:biorthogonal_lem1}
\end{lemma}
\begin{theorem}[Complementary Filters {\cite[prop.~3.13]{Vetterli_book}}]
Given a polynomial kernel $\hat h_0(l)$, there exists a complementary polynomial kernel $\hat g_0(l)$ which satisfies the 
perfect reconstruction relation in (\ref{eq:biorthogonal_perfect_reconstruct}), if and only if $\hat h_0(1 + l)$ and 
$\hat h_0(1 - l)$ are coprime. 
 \label{thm:biorthogonal_thm1}
\end{theorem}
\begin{proof}
Let us denote $a(l) = h_0(1 + l)$, $b(l) = h_0(1 - l) = a(-l)$, 
and $c(l) = 2$. Then,~(\ref{eq:biorthogonal_perfect_reconstruct}) can be written 
in the same form as (\ref{eq:bezout}), i.e.,
\begin{equation}
a(l) x(l)  + a(-l) y(l)  = 2,
 \label{eq:bezout1}
\end{equation} 
We first note  that if a polynomial solution $[x(l),~y(l)]$  of~(\ref{eq:bezout1}) exists, 
then
\begin{equation}
a(-l) x(-l)  + a(l) y(-l)  = 2,
 \label{eq:bezout2}
\end{equation}
also has a polynomial solution. Subsequently, combining~(\ref{eq:bezout1}) 
and~(\ref{eq:bezout2})
and choosing $\hat g_0(1 + l) = 1/2(x(l) + y(-l))$, we find that 
\begin{equation}
a(-l) g_0(1+l)  + a(l) g_0(1-l)  = 2,
 \label{eq:bezout3}
\end{equation}
also has a polynomial solution. 
However, based on Lemma \ref{lem:biorthogonal_lem1}, 
(\ref{eq:bezout1}) has 
a polynomial solution 
{\em if and only if}  $gcd( h_0(1 + l),h_0(1 - l))$ 
divides $c(l) = 2$, which is 
a prime number. 
This implies $gcd( h_0(1 + l),h_0(1 - l))$ is either $1$ or $2$ for all $l \in [-1~1]$, 
which is true iff $\hat h_0(1+l)$ and $\hat h_0(1-l)$ do not have any common roots. This implies that $\hat h_0(1+l)$ and $\hat h_0(1-l)$ are coprime.
%
\end{proof}

\begin{corollary}[{\cite[exercise.~3.12]{Vetterli_book}}]
 There is always a complementary filter for the polynomial kernel $(1 + l)^k$, i.e.,
\begin{equation}
(1 + l)^kR(l) + (1 - l)^kR(-l) = 2
 \label{eq:binomial_filter}
\end{equation}
always has a real polynomial solution $R(l)$ for $k \geq 0$.
\label{corr:corr1}
\end{corollary}
\begin{proof}
 Let us denote $a(l) = (1 + l)^k$, $b(l) = (1 - l)^k$, $x(l) = R(l)$, $y(l) = R(- l)$ 
and $c(l) = 2$. Then,~(\ref{eq:binomial_filter}) can be written in the same form as (\ref{eq:bezout}). Since $a(l)$ and $b(l)$, in this case 
are coprime, therefore $gcd( a(l),b(l)) =1$ divides $c(l) = 2$. Hence, a polynomial $R(l)$, which satisfies~(\ref{eq:binomial_filter}) always exists. 
\end{proof}


For a perfect reconstruction biorthogonal filterbank, we need to design a polynomial half-band kernel $\hat p(\lambda)$ that 
satisfies~(\ref{eq:biorthogonal_perfect_reconstruct}) for all $\lambda \in [0~2]$, or  equivalently $\hat p(l)$ that satisfies~(\ref{eq:p_lambda_cv}) for all 
$l \in [-1~1]$.
Following Daubechies' approach \cite{CDF9}, 
we propose 
a {\em maximally-flat} design, in which  we assign 
$K$ roots to $\hat p(\lambda)$ at the lowest eigenvalue 
(i.e., at $\lambda = 0$). Subsequently, we select $\hat p(\lambda)$ to be the shortest length
polynomial, which has $K$ roots at $\lambda = 0$ and satisfies (\ref{eq:p_lambda_cv}). 
This implies that 
$\hat p(1 + l)$ has
$K$ roots
at $l = -1$,  
and can be expanded as:
\begin{equation}
 \hat p(1+l)  =  (1 + l)^{K} \underbrace{\sum_{m = 0}^{k}r_ml^m.}_{R(l)} 
 \label{eq:p_lambda_expansion1}
\end{equation}
where $R(l)$ is the residual $k$ degree polynomial of $l$. By Corollary \ref{corr:corr1}, there always exist such a polynomial 
$R(l)$.
On the other hand, Proposition \ref{prop:odd_p_lambda} says
that any $\hat p(1 + l)$ that satisfies~(\ref{eq:p_lambda_cv}) has to be an odd-degree polynomial.
Hence,
$\hat p(1 + l)$ can also be expanded as:
\begin{equation}
 \hat p(1+l) = 1 + \sum_{n = 0}^{M} c_{2n+1}l^{2n+1}. 
 \label{eq:p_lambda_expansion2}
\end{equation}
for a given $M$.
%
Comparing  (\ref{eq:p_lambda_expansion1}) and (\ref{eq:p_lambda_expansion2}), we get:
\begin{equation}
\displaystyle (1 + l)^{K}\sum_{m = 0}^{k}r_ml^m =  1 + \sum_{n = 0}^{M} c_{2n+1}l^{2n+1}.
 \label{eq:p_lambda_expnasion_comp}
\end{equation}
Comparing the constant terms in the left and right side of (\ref{eq:p_lambda_expnasion_comp}), we get $r_0 = 1$.
Further, comparing the highest powers on both sides of (\ref{eq:p_lambda_expnasion_comp}) we get: 
\begin{equation}
 M = \frac{K+k-1}{2}  
 \label{eq:highest_power_compare}
\end{equation}
Further, the right side in (\ref{eq:p_lambda_expnasion_comp}) has $M$ constraints 
$c_{2n} = 0$ for $n = \{1,2,...K\}$, and the left side in (\ref{eq:p_lambda_expnasion_comp}) 
has $k$ unknowns $r_m$ for $m = \{1,2,...k\}$. 
In order to get a unique $\hat p(1+l)$ that satisfies (\ref{eq:p_lambda_cv}), 
we must have equal number of unknowns and constraints, i.e,
\begin{equation}
 M = k =  \frac{K + k-1}{2}~~\Rightarrow~~M = K - 1.   
 \label{eq:biorthogonal51}
\end{equation}
Thus, (\ref{eq:p_lambda_expnasion_comp}) can be written as: 
\begin{equation}
\displaystyle  (1 + l)^{K}( 1+ \sum_{m = 1}^{K-1}r_ml^m) = 1 + \sum_{n = 0}^{K-1} c_{2n+1}l^{2n+1},
 \label{eq:p_lambda_expansion_comp_final}
\end{equation}
and $K-1$ unknowns can be found uniquely, by 
solving a linear system  of $K-1$ equations. Note that given $K$, the length of $\hat p(l)$ (i.e, highest degree) 
is $K+M = 2K-1$. 
As an example, we design $\hat p(\lambda)$ with $K = 2$ zeros at $l = 0$. 
In this case $\hat p(1+l)$ can be written as:
\begin{equation}
\displaystyle \hat p(1+l) =  (1 + l)^2 (1 + r_1l) = 1 + (r_1 + 2)l + (1 + 2r_1) l^2 + r_1 l^3 \nonumber \\
 \label{eq:p_lambda_ex1}
\end{equation}
Since $\hat p(1+l)$ is an odd polynomial, the term corresponding to $l^2$ is zero, i.e.,  $1 + 2r_1 = 0$ or $r_1 = -1/2$. Therefore, 
$\hat p(1+l)$ is given as:
\begin{eqnarray}
\hat p(1 + l) &=&  (1 + l)^2( 1 - \frac{1}{2}l),
 \label{eq:p_lambda_ex2}
\end{eqnarray}
which implies that:
\begin{eqnarray}
p(\lambda) &=&  \frac{1}{2}\lambda^2( 3 - \lambda).
 \label{eq:p_lambda_ex3}
\end{eqnarray}
In Figure~\ref{fig:p_lambda_plot}, we plot $\hat p(\lambda)$ for various values of $K$, and it can be 
seen that by increasing $K$, we get a $\hat p(\lambda)$ representing a better approximation to the ideal halfband filter. 
\begin{figure}[htb]
\centering
\includegraphics[width=5in]{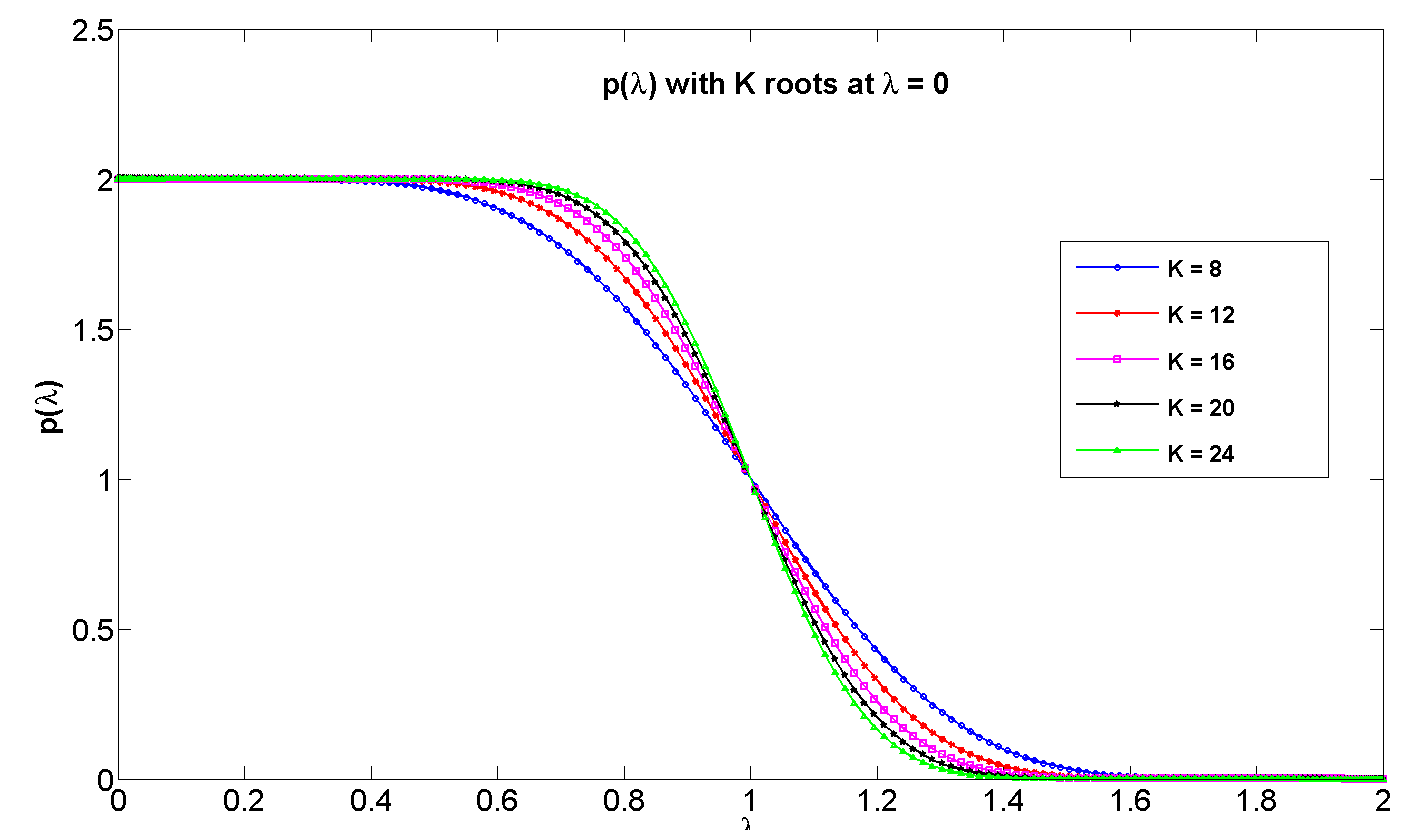}
\caption{ The spectral distribution of $\hat p(\lambda)$ with $K$ zeros at $\lambda = 0$} 
\label{fig:p_lambda_plot}
\end{figure}

Note that the graph-QMF designs in~\cite{SunilTSP} were based on selecting 
$\hat g_0(\lambda) = h_0(\lambda)$
hence $p_{QMF}(\lambda) = h_0^2(\lambda)$.  Thus, if 
$\hat h_0(\lambda)$ is a polynomial kernel 
then $p_{QMF}(\lambda)$ is 
the square of a polynomial, and therefore should have an even degree.
However, as proven in Proposition~\ref{prop:odd_p_lambda}, 
any polynomial $\hat p(\lambda)$ 
which satisfies~(\ref{eq:biorthogonal_perfect_reconstruct}) 
is an odd-degree polynomial. 
{\bf Therefore, $\hat h_0(\lambda)$ 
in the graph-QMF designs, cannot be an exact polynomial.} %
\subsection{Spectral factorization of half-band kernel $\hat p(\lambda)$}
Once we obtain $\hat p(\lambda)$ by using the above mentioned design, we 
need to factor it into filter kernels $\hat h_0(\lambda)$ and $\hat g_0(\lambda)$. 
Since $\hat p(\lambda)$ is a real 
polynomial of odd degree, it has at least  
one real root and all the complex roots occur in conjugate pairs. 
Since 
we want the two kernels to be polynomials with real coefficients, 
each complex conjugate root pair of $\hat p(\lambda)$ should be assigned together to either 
$\hat h_0(\lambda)$ or $\hat g_0( \lambda)$. While any such factorization 
would lead to perfect reconstruction biorthogonal filterbanks, 
of particular interest is the design of filterbanks that 
are as close to orthogonal as possible. 
For this, we define a criterion based on energy preservation.
In particular, we compute the 
Riesz bounds of analysis wavelet 
transform $\Tm_a$, which are the 
tightest lower and upper bounds, 
$A>0$ and $B < \infty$, 
of $||\Tm_a \fv||_{2}$, for any 
graph-signal $\fv$ with $||\fv||_2 = 1$. 
For near-orthogonality, 
we require $A \approx B \approx 1$. 
The bounds  $A$ and $B$ can be computed 
as the minimum and maximum singular 
values of the overall analysis side 
transform $\Tm_a$ of 
the two channel filterbank. The transform 
matrix $\Tm_a$ consists of some rows of lowpass transform $\Hm_0$ 
(those corresponding to the $L$ set) and some rows of 
highpass transform $\Hm_1$ (those corresponding to the $H$ set) and
can be written as: 
\begin{equation}
\Tm_a = \frac{1}{2}(\Id + \Jm_{\beta}) \Hm_0 + \frac{1}{2}(\Id - \Jm_{\beta}) \Hm_1 
 \label{eq:analysis_tx}
\end{equation}
The singular values of $\Tm_a$ are also the square roots of eigenvalues of $\Tm_a^\top \Tm_a$, which 
can be expanded as (see~\cite{SunilTSP} for details):
%
%
%
%
%
\begin{eqnarray}
\Tm_{a}^\top\Tm_a & = & 1/2\sum_{\lambda \in \sigma(\Bc)}\underbrace{(h_0^2(\lambda)+h_1^2(\lambda))}_{C(\lambda)}\Pm_{\lambda}  \nonumber \\
& + & 1/2 \sum_{\lambda \in \sigma(\Bc)}\underbrace{(h_1(\lambda)h_1(2 - \lambda)-h_0(\lambda)h_0(2 - \lambda))}_{D(\lambda)}\Jm_{\beta}\Pm_{\lambda}, 
\label{eq:biorthogonal_tx_orthogonal}
\end{eqnarray}
where $\Jm_{\beta}$ is the diagonal matrix of binary function $\beta$.
In~(\ref{eq:biorthogonal_tx_orthogonal}), the term $D(\lambda)$ consists of product terms 
$\hat h_0(\lambda) \hat h_0(2 - \lambda)$ and $\hat h_1(\lambda) \hat h_1(2 - \lambda)$, 
which are small for $\lambda$ away from the transition band around 
$1$ (since these are the products of a low pass and a high pass kernel). Further, in the transition band
when $\lambda$ is close to $1$, the value of $D(\lambda) \approx \hat h_0^2(\lambda)-\hat h_1^2(\lambda)$ is very small 
compared to $C(\lambda) \approx \hat h_0^2(\lambda) + \hat h_1^2(\lambda)$.
%
%
Therefore, we can ignore $D(\lambda)$ in comparison to 
$C(\lambda)$, and (\ref{eq:biorthogonal_tx_orthogonal}) can be approximately reduced to:
{\small 
\begin{equation}
\Tm_{a}^\top\Tm_a \approx  1/2\sum_{\lambda \in \sigma(\Bc)}\underbrace{(h_0^2(\lambda)+h_1^2(\lambda))}_{C(\lambda)}\Pm_{\lambda}  
\label{eq:tx_orthogonal3}
\end{equation}
}
Thus, $\Tm_{a}^\top\Tm_a$ is a spectral transform with eigenvalues $1/2(\hat h_0^2(\lambda)+\hat h_1^2(\lambda))$ for $\lambda \in \sigma(\Bc)$, 
and the Riesz Bounds can be given as:
{\small 
\begin{eqnarray}
A & = & \sqrt{\inf_{\lambda} \frac{1}{2}(\hat h_0^2(\lambda)+ \hat h_1^2(\lambda))} \nonumber \\
B & = & \sqrt{\sup_{\lambda} \frac{1}{2}(\hat h_0^2(\lambda)+\hat h_1^2(\lambda))} 
\label{eq:Riesz_bound}
\end{eqnarray}
}  
We define $\Theta$, as the measure of 
orthogonality, given as:
\begin{equation}
 \Theta = 1 - \frac{|B-A|}{|B+A|}.
 \label{eq:Theta_def}
\end{equation}
For orthogonal filterbanks $\Theta = 1$. We choose filters with least dissimilar lengths, and for near orthogonal 
designs, compute $\Theta$ for all such possible factorizations (there are ${2K-1 \choose K}$ possible choices), and 
choose the factorization with the maximum absolute value of $\Theta$. 
Note that the computation of $\Theta$, and hence the choice of the best solution depends on the exact 
distribution of eigenvalues $\lambda$ in the interval $[0~2]$, which in turn depends on the underlying graph. 
For a graph independent design, 
%
we approximate $A^2$ and $B^2$ as the 
lowest and highest values, respectively,
of $1/2(h_0^2(\lambda)+h_1^2(\lambda))$ at 
$100$ uniformly sampled points from the continuous 
region $[0~2]$, respectively, and use these approximations to compute $\Theta$.

\subsection{Unity gain compensation}
\label{sec:GC}
In order to 
avoid unnecessary growth of 
dynamic range in the 
output, it is desirable to  
normalize the filterbanks, 
such that the {\em impulse responses} of 
lowpass and highpass filters have equal (ideally unity) gain. 
In the graph-QMF case~\cite{SunilTSP}, the 
orthogonality condition ensures that
all filters have unity gain. However, this is not true for the
proposed graphBior filterbanks. 
Similar {\em gain compensations} 
have also been proposed for biorthogonal DWT filterbanks 
(for example, see~\cite{JPEG2000Compensate}). 
The JPEG2000 standard, for example, requires 
the lowpass filter to have unity response for the
DC frequency ($\omega = 0$), 
and the highpass filter to have unity response 
for the Nyquist frequency ($\omega = \pi$).  In this paper, 
we follow similar specifications, as given in JPEG2000, to normalize the gains of 
graphBior filters. In the bipartite graph case, $\omega = 0$ 
and $\omega = \pi$ correspond to 
the lowest ($\lambda = 0$) and the highest ($\lambda = 2$) magnitude eigenvalues, 
respectively.
Thus, given a filter $\Hm$ with spectral kernel $h(\lambda)$, the gain factor 
of $\Hm$ is equal to $|1 / h(0)|$, if $\Hm$ is a lowpass filter, and is equal to 
$|1 / h(2)|$, if $\Hm$ is a highpass filter.  In the filterbank implementation,
a GC block is applied at the analysis side 
in each 
channel after filtering and 
downsampling, and an inverse GC 
block is applied at the synthesis 
side prior to filtering and upsampling. As a result, the filterbank remains perfect reconstruction. 

\subsection{Nomenclature and design of graphBior filterbanks}
\label{sec:biorthgonal_spectral_fb}
The proposed biorthogonal filterbanks are specified by four parameters ($k_0,~k_1,~l_0,~l_1$), where
$k_0$ is the number of roots of low pass analysis kernel $\hat h_0(\lambda)$ at $\lambda = 0$, 
$k_1$ is the number of roots of low pass synthesis kernel $\hat g_0(\lambda)$ at $\lambda = 0$, 
$l_0$ is the highest degree of low pass analysis kernel $\hat h_0(\lambda)$, 
and $l_1$ is the highest degree of low pass synthesis kernel $\hat g_0(\lambda)$. 
The other two filters, namely $\hat h_1(\lambda)$ and $\hat g_1(\lambda)$, can be computed as in~(\ref{eq:perfect_reconstruct2}).
Given these specifications, we design $\hat p(\lambda) = \hat h_0(\lambda) \hat g_0(\lambda)$ 
as a maximally flat half band polynomial kernel with $K = k_0 + k_1$  roots 
at $\lambda = 0$. As a result, $\hat p(\lambda)$ turns out to be a $2K-1$ degree polynomial, 
and we factorize it into $\hat h_0(\lambda)$ 
and $\hat g_0(\lambda)$, with least dissimilar lengths (i.e., we choose $l_0 = K$ and $l_1 = K-1$). 
We use  $\Theta$ as the criterion to compare various possible factorizations, and choose 
the one with the maximum value of $\Theta$. This leads to a unique design of biothogonal filterbanks. 
We term our proposed filterbanks as {\em graphBior$(k_0,k_1)$}. We designed graphBior filterbanks for various 
values of $(k_0,k_1)$, and we observed that designs with $k_0 = k_1$ stand out, 
as they are close to orthogonal and have near-flat pass-band responses. 
The lowpass and highpass analysis kernels are plotted in Figure~\ref{fig:spectral_response}, and their coefficients are shown in 
Table~\ref{tab:graph_Bior_coeffs}. 
\begin{figure}[htb]
\begin{center}
\subfigure[]{
   \includegraphics[width = 2.5in] {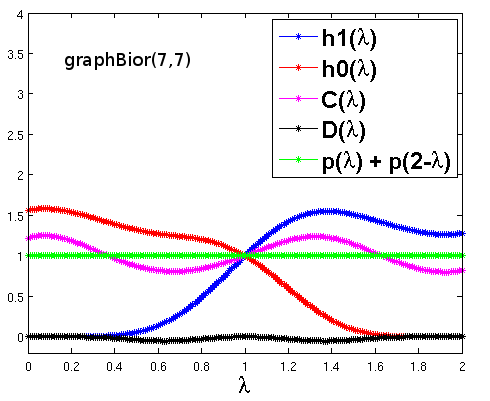}
   \label{fig:graphBior7}
 }
\subfigure[]{
   \includegraphics[width = 2.5in] {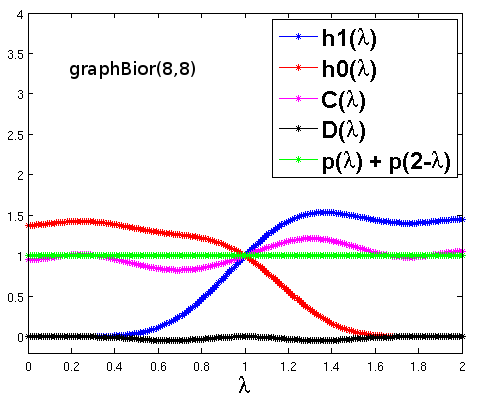}
   \label{fig:graphBior8}
 }
\subfigure[]{
   \includegraphics[width = 2.5in] {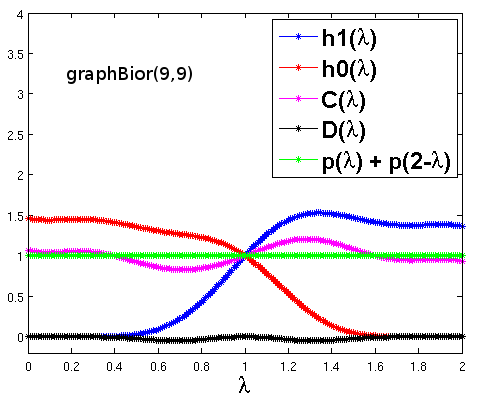}
   \label{fig:graphBior9}
 }
\subfigure[]{
   \includegraphics[width = 2.5in] {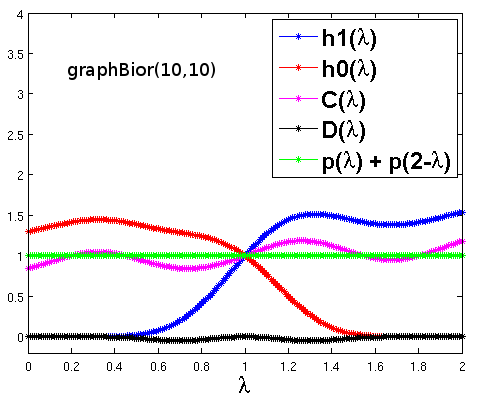}
   \label{fig:graphBior10}
 }
\subfigure[]{
   \includegraphics[width = 2.5in] {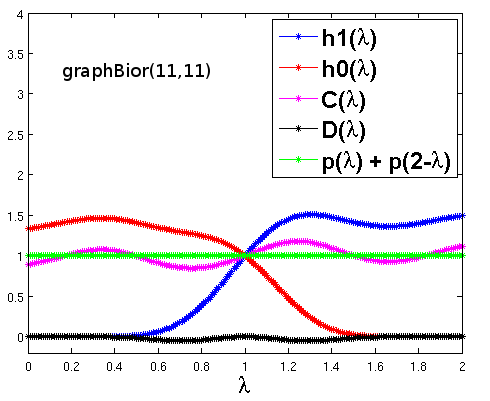}
   \label{fig:graphBior11}
 }
\subfigure[]{
   \includegraphics[width = 2.5in] {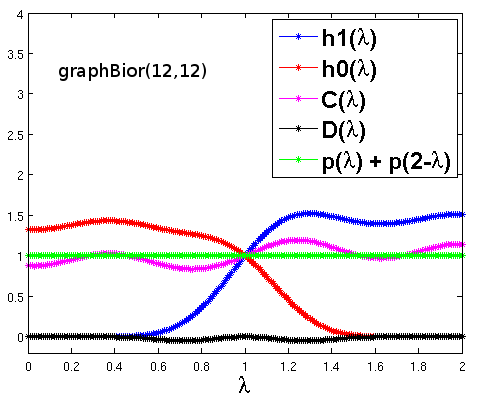}
   \label{fig:graphBior12}
 }
\caption{ Spectral responses of graphBior$(k_0,k_1)$ filters on a bipartite graph. In each plot, $\hat h_0(\lambda)$ and 
$\hat h_1(\lambda)$ are lowpass and highpass analysis kernels, $C(\lambda)$ and $D(\lambda)$ constitute the spectral response of the overall 
analysis 
filter $\Tm_a$, as in~(\ref{eq:biorthogonal_tx_orthogonal}). For near-orthogonality $D(\lambda) \approx 0$ and $C(\lambda) \approx 1$. Finally, 
$(p(\lambda) + p(2-\lambda))/2$ represents perfect reconstruction property as in~(\ref{eq:p_lambda_cv}), 
and should be constant equal to $1$, for perfect reconstruction.}
\label{fig:spectral_response}
\end{center}
\end{figure}

\begin{table}[htbp]
\begin{center}
\begin{tabular}{|c|p{9cm}|}
\hline
\textbf{graphBior($k_0$,$k_1$)} & \multicolumn{1}{c|}{\textbf{filter coefficients }} \\ \hline
$k_0$ = 6, $k_1$ = 6 & $\hat h_1$ = [-0.3864    4.0351  -17.0630   36.5763  -39.8098   17.6477         0         0         0         0         0         0] \\
                                      & $\hat h_0$ = [ 0.4352   -4.9802   23.2396  -55.4662   67.2657  -29.0402  -13.0400    7.5253    9.5267   -4.8746   -2.0616    1.2633    1.2071] \\ \hline
$k_0$ = 7, $k_1$ = 7 & $\hat h_1$ = [0.3115   -3.9523   21.0540  -60.3094   98.0605  -85.9222   31.7578         0         0         0         0         0         0          0] \\
                                      & $\hat h_0$ = [ -0.4975    6.8084  -39.6151  126.2423 -234.3683  241.5031  -97.6557  -46.2635   62.1232  -19.3648   -2.0766    6.5886   -4.5632    0.5775    1.5614] \\ \hline
$k_0$ = 8, $k_1$ = 8 & $\hat h_1$ = [-0.3232    4.7284  -29.7443  104.3985 -221.0705  282.7915 -202.6283   62.8477         0         0         0         0         0         0         0         0] \\
                                      & $\hat h_0$ = [  0.4470   -6.9872   47.5460 -183.6940  440.0924 -670.0905  643.3979 -396.0713  209.9824 -154.0976   92.8617  -30.8228   16.6112  -12.7664    3.2403   -0.0284    1.3793] \\ \hline
\end{tabular}
\end{center}
\caption{ Polynomial expansion coefficients (highest degree first) of graphBior $(k_0,k_1)$ filters (approximated to $4$ decimal places) on a bipartite graph. Refer to the Matlab code for more accurate coefficients. }
\label{tab:graph_Bior_coeffs}
\end{table}

\section{ zeroDC graphBior Filterbanks}
\label{sec:zerodc_gb}
In many application such as images, videos and wireless sensor networks etc., 
where the underlying graph resides in a physical space, 
an all constant 
signal (a DC signal) has a physical interpretation, and the 
wavelet filters are designed to be orthogonal 
to the dc signal. In our proposed nonzeroDC graphBior design, 
wavelet transforms (highpass transforms) are designed to be 
orthogonal to the eigenvector corresponding to $0$ eigenvalue of 
the normalized Laplacian matrix $\Lcb$, 
which is $\Dm^{1/2}\onev$ where $\Dm$ is the degree matrix. 
If the graph is almost regular (i.e., has almost same degree at all nodes), 
then the vector $\Dm^{1/2}\onev$ is almost constant. 
%
%
%
To obtain a zero DC response for other non-regular graphs
we propose {\em zeroDC graphBior filterbanks} designs.
Construction-wise,  the only difference between  
{\em zeroDC graphBior filterbanks} and {\em nonzeroDC graphBior filterbanks}
is that the former are designed using random-walk Laplacian matrix $\Lcb_r$
while the latter are designed using normalized Laplacian matrix $\Lcb$. 
The two Laplacian matrices are {\em similar}, hence 
their eigenvalues are identical. Therefore, no change is needed in the design 
if spectral kernels designed above. 
%
In order to compute zeroDC filterbanks, we simply 
replace $\Lcb$ by $\Lcb_r$ in~(\ref{eq:spectral_tx}), i.e.,
\begin{eqnarray}
 \Hm_{ri} &=& \hat h_i(\Lcb_r) \nonumber \\
 \Gm_{ri} &=& \hat g_i(\Lcb_r),
 \label{eq:asym_filters}
\end{eqnarray}
and replacing symmetric filters $\Hm_i$ and $\Gm_i$ by $\Hm_{ri}$ and $\Gm_{ri}$, respectively,  
in the two-channel nonzeroDC filterbank implementation shown in Figure~\ref{fig:graph_filterbank}. 
%
The following results
describe the properties of proposed zeroDC filterbanks:

\begin{proposition}[{\bf zero DC response}]
For any connected graph, if the spectral kernel $\hat h(\lambda)$ 
is such that $\hat h(0) = 0$, then the transform $\Hm_{r} = \hat h(\Lcb_r)$ has a zero DC response, i.e., $\Hm_r\onev = \zerov$.
Further, if $\hat h(\lambda)$ is a polynomial kernel then  the transforms $\Hm_r=\hat h(\Lcb_r)$ and $\Hm = \hat h(\Lcb)$ 
are related as:
 \begin{eqnarray}
 \displaystyle {\Hm}_{r} = \Dm^{-1/2}\Hm\Dm^{1/2}.
\label{eq:biorthogonal_spectral_tx1}
\end{eqnarray} 
\end{proposition}
\begin{proof}
The random walk Laplacian matrix $\Lcb_r$ can be diagonalized as:
\begin{equation}
 \Lcb_r = \Dm^{-1/2}\Um \Lam (\Dm^{-1/2}\Um)^{-1} =  \Dm^{-1/2}\Um \Lam\Um^\top\Dm^{1/2}. \nonumber
\end{equation}
Therefore, any function of $\Lcb_r$ can be written as:
\begin{equation}
 \Hm_r = \hat h(\Lcb_r) = \Dm^{-1/2}\Um \hat h(\Lam) \Um^\top\Dm^{1/2} =  \Dm^{-1/2}\hat h(\Lcb)\Dm^{1/2} = \Dm^{-1/2} \Hm \Dm^{1/2}. \nonumber
\end{equation}
Further, if $\uv_l$ is an eigenvector of normalized Laplacian matrix $\Lcb$ 
corresponding to eigenvalue $\lambda_l$, 
then by definition $\Hm\uv_l = \hat h(\lambda_l)\uv_l$. Since $\uv_0 = \Dm^{1/2}\onev$
is the eigenvector of $\Lcb$ with eigenvalue $0$, this implies:
\begin{equation}
 \Hm_r\onev = \Dm^{-1/2} \Hm \Dm^{1/2}\onev = \lambda_0 \onev = \zerov. \nonumber
\end{equation}
Thus $\Hm_r$ has zero DC response. 
\end{proof}
\begin{proposition}[{\bf perfect reconstruction property}]
The zeroDC filterbanks designed using graphBior spectral kernels are also perfect reconstruction. 
\end{proposition}
\begin{proof}
Similar to (\ref{eq:overall_tx}), 
the overall transfer function of the zeroDC filterbank can be written as:
\begin{eqnarray}
 \displaystyle \hat \fv  &=&   \frac{1}{2}\Gm_{r0}(\Id + \Jm_{\beta})\Hm_{r0}\fv + \frac{1}{2}\Gm_{r1}(\Id - \Jm_{\beta})\Hm_{r1}\fv \nonumber \\
& = & \frac{1}{2}(\Gm_{r0}\Hm_{r0} + \Gm_{r1}\Hm_{r1})\fv +  \frac{1}{2}(\Gm_{r0}\Jm_{\beta}\Hm_{r0} - \Gm_{r1}\Jm_{\beta}\Hm_{r1})\fv.
\label{eq:overall_tx_raw_asym}
\end{eqnarray}
Using the similarity relation given in~(\ref{eq:biorthogonal_spectral_tx1}), we can simplify~(\ref{eq:overall_tx_raw_asym}) as:
\begin{eqnarray}
 \displaystyle \hat \fv  &=& \frac{1}{2}(\Dm^{-1/2}\Gm_{0}\Dm^{1/2}\Dm^{-1/2}\Hm_{0}\Dm^{1/2} + \Dm^{-1/2}\Gm_{1}\Dm^{1/2}\Dm^{-1/2}\Hm_{1}\Dm^{1/2})\fv \nonumber \\
&+&  \frac{1}{2}(\Dm^{-1/2}\Gm_{0}\Dm^{1/2}\Jm_{\beta}\Dm^{-1/2}\Hm_{0}\Dm^{1/2} - \Dm^{-1/2}\Gm_{1}\Dm^{1/2}\Jm_{\beta}\Dm^{-1/2}\Hm_{1}\Dm^{1/2})\fv.
\label{eq:overall_tx_asym_raw}
\end{eqnarray}
In~(\ref{eq:overall_tx_asym_raw}), the matrices $\Dm^{1/2},\Jm_{\beta}$, and $\Dm^{-1/2}$ are diagonal matrices and hence commute with each other. Therefore, 
\begin{equation}
\Dm^{1/2}\Jm_{\beta}\Dm^{-1/2} =  \Jm_{\beta}\Dm^{1/2}\Dm^{-1/2} = \Jm_{\beta}
\end{equation}
Thus,~(\ref{eq:overall_tx_asym_raw}), can be simplified as:
\begin{eqnarray}
 \displaystyle \hat \fv  &=& \frac{1}{2}(\Dm^{-1/2}\Gm_{0}\Hm_{0}\Dm^{1/2} + \Dm^{-1/2}\Gm_{1}\Hm_{1}\Dm^{1/2})\fv \nonumber \\
&+&  \frac{1}{2}(\Dm^{-1/2}\Gm_{0}\Jm_{\beta}\Hm_{0}\Dm^{1/2} - \Dm^{-1/2}\Gm_{1}\Jm_{\beta}\Hm_{1}\Dm^{1/2})\fv \nonumber \\
& =& \Dm^{-1/2} \Tm_{eq} \Dm^{1/2}\fv + \Dm^{-1/2} \Tm_{alias} \Dm^{1/2}\fv \nonumber \\
& = & \Dm^{-1/2} (\Tm_{eq}  + \Tm_{alias}) \Dm^{1/2}\fv,
\label{eq:overall_tx_asym}
\end{eqnarray}
where $\Tm_{eq}$ and $\Tm_{alias}$ correspond to the overall transfer function of nonzeroDC filterbanks, as defined in
(\ref{eq:overall_tx}). {\em Therefore, the zeroDC filterbank implementation is equivalent to pre-multiplying the input by 
$\Dm^{-1/2}$ and post-multiplying the output by $\Dm^{1/2}$, and  if the nonzeroDC filterbank
is PR (i.e., 
$\Tm_{eq}  + \Tm_{alias} = c\Id$ ) then the corresponding zeroDC filterbank is also PR.}
\end{proof}
\begin{proposition}[{\bf Riesz bounds}]
The zeroDC filterbanks form a Riesz basis with lower bound $A\sqrt{d_{min}/d_{max}}$ and upper bound $B\sqrt{d_{max}/d_{min}}$, 
where $A$ and $B$ are the lower and upper bounds of the Riesz basis formed by corresponding nonzeroDC graphBior filterbanks.
%
\end{proposition}
\begin{proof}
Referring to Figure~\ref{fig:graph_filterbank}, 
the wavelet coefficient vector $\wv$ produced in the zeroDC filterbanks can be written as:
\begin{eqnarray}
\wv_r = \Tm_{ra}\fv &= & \frac{1}{2}(\Id - \Jm_{\beta})\Hm_{r0}\fv + \frac{1}{2}(\Id + \Jm_{\beta})\Hm_{r1}\fv \nonumber \\
& = & \frac{1}{2}(\Hm_{r1} + \Hm_{r0})\fv + \frac{1}{2}\Jm_{\beta}(\Hm_{r1} - \Hm_{r0})\fv \nonumber \\
& = & \frac{1}{2}\Dm^{-1/2}(\Hm_{1} + \Hm_{0})\Dm^{1/2}\fv + \frac{1}{2}\Dm^{-1/2}\Jm_{\beta}(\Hm_{1} - \Hm_{0})\Dm^{1/2}\fv \nonumber \\
& = & \Dm^{-1/2}\Tm_a\Dm^{1/2}\fv
 \label{eq:analysis_tx_asym}
\end{eqnarray}
This implies that the $n^{th}$ output can we written as:
\begin{equation}
 w_r[n] = \sum_{m=1}^N \sqrt{\frac{d_m}{d_n}}T_a(n,m)f[m]
 \label{eq:nth_outp_asym}
\end{equation}

Note that if the graph is almost regular, i.e., $\frac{d_m}{d_n} \approx 1$, then $w_r[n] \approx \sum_{m=1}^N T_a(n,m)f[m] = w[n]$, where 
$w[n]$ is the $n^{th}$ output of the corresponding nonzeroDC filterbank. 
In order to obtain a worst-case bound, if we define $\fv_{\Dm} = \Dm^{1/2}\fv$, and $\wv_{\Dm} = \Dm^{1/2}\wv$, then~(\ref{eq:analysis_tx_asym}) 
can be written as $\wv_{\Dm} = \Tm_a\fv_{\Dm}$. Thus, if the corresponding nonzeroDC filterbank 
is biorthogonal with Riesz bounds $A$ and $B$, 
then $A||\fv_{\Dm}|| \leq ||\wv_{\Dm}|| \leq B||\fv_{\Dm}||$ (the $2$-norm). However,
\begin{eqnarray}
 d_{min} \sum_{i=1}^N w^2(i) \leq ||\wv_{\Dm}||^2 &=& \sum_{i=1}^N d_i w^2(i) \leq d_{max} \sum_{i=1}^N w^2(i) \nonumber \\
d_{min} \sum_{i=1}^N f^2(i) \leq ||\fv_{\Dm}||^2 &=& \sum_{i=1}^N d_i f^2(i) \leq d_{max} \sum_{i=1}^N f^2(i),
\label{eq:orthogonal_asym}
\end{eqnarray}
where $d_{min}$ is the minimum degree in the graph ($1$ if there is an isolated node), and $d_{max}$ is the maximum degree.
Using~(\ref{eq:orthogonal_asym}), we obtain: 
\begin{eqnarray}
 d_{min}||\wv||^2 \leq ||\wv_{\Dm}||^2 &\leq& B^2||\fv_{\Dm}||^2 \leq B^2d_{max} ||\fv||^2 \nonumber \\
d_{min}A^2||\fv||^2 \leq A^2||\fv_{\Dm}||^2 &\leq& ||\wv_{\Dm}||^2 \leq d_{max} ||\wv||^2,
\end{eqnarray}
and 
\begin{equation}
\left(A\frac{d_{min}}{d_{max}}\right)||\fv||^2  \leq ||\wv||^2 \leq \left(B\frac{d_{max}}{d_{min}}\right)||\fv||^2 
\label{eq:frame_asym}
\end{equation}
Thus, the zero graphBior filterbanks defines a Riesz basis in the graph-signal space,  with lower bound $A_r = A\sqrt{d_{min}/d_{max}}$ and upper-bound 
$B_r = B\sqrt{d_{max}/d_{min}}$. 
\end{proof}
Note that for regular graphs $d_{min} = d_{max}$, hence $\{A_r,B_r\} = \{A,B\}$. However, for irregular graphs 
the measure of orthogonality $\Delta_r = A_r/B_r = (d_{min}/d_{max}) \Theta$ tend to be smaller than $\Theta$, which implies that the basis functions in 
zeroDC filterbanks are more coherent than the basis functions in nonzeroDC filterbanks. This is 
also confirmed empirically in Table~\ref{tab:SNR_compare}.

The decision of whether to use zeroDC graphBior or 
nonzeroDC graphBior filterbanks, depends 
upon the interpretation of an all constant signal $\onev$ 
and its degree normalized form $\Dm^{1/2}\onev$ 
in the context of the problem. 
For example, in graphs arising from
physical domains
(sensor networks, 
transport networks, images and videos etc.), 
the graph signals 
are often  nearly constant (or piecewise constant). 
In these cases, the all-constant signal 
should be preserved as the lowpass 
signal, and therefore the 
zeroDC filterbanks should be preferred over the nonzeroDC filterbanks. 
On the other hand, it is shown in~\cite{Mihail'02}
that for highly irregular graphs (such as online social networks, Internet etc.)
the spectral analysis gets influenced by the presence of 
high degree nodes, and thus misses the structure 
around low degree nodes. Therefore, the degree 
normalized nonzeroDC filterbanks should be used for these cases.
All the examples presented in Section~\ref{sec:experiments} belong 
to the former category (i.e., they arise in physical domains). Therefore,
the zeroDC filterbanks are found to perform better than the nonzeroDC filterbanks.

\section{Multi-dimensional and multi-resolution implementations}
\label{sec:multi_dimensional}
So far we have described how to implement graphBior filterbanks 
on bipartite graphs. This is because bipartite graphs 
provide perfect reconstruction conditions in terms of simple 
conditions on spectral responses in these filterbanks.
However,
not all graphs are bipartite. 
For arbitrary graphs, we proposed 
in~\cite{SunilTSP,ICASSP12Sunil} 
to decompose the graph $G$ 
into $K$ link-disjoint bipartite
subgraphs, each defined on 
the entire set of vertices and 
their union 
covering almost all of the 
links in the graph. 
Consequently, we 
implemented filtering/downsampling 
operation in $K$ stages,
restricting the operations in 
each stage to only 
one bipartite graph. 
An example of $2$-dimensional 
bipartite subgraph 
decomposition 
is shown
in Figure~\ref{fig:2D_example_new}, in which 
the graph $G$
is divided into 
$4$ clusters $LL,LH,HL$ and $HH$. 
The first bipartite graph $B_1$ corresponds to 
partitions $L1 = LL \cup LH$ and $H1 = HL \cup HH$, 
and all the links connecting nodes in the two partitions. 
Subsequently, these links are removed 
from $G$ and the second bipartite subgraph $B_2$ corresponds to 
partitions $L2 = HL \cup LL$ and $H2 = LH \cup HH$, and  all the links between $L2$ and $H2$
from the remaining set of links. The remaining links are either discarded 
or used to further compute third and fourth bipartite subgraphs etc. 
\begin{figure}[htb]
\centering
\includegraphics[width=5in]{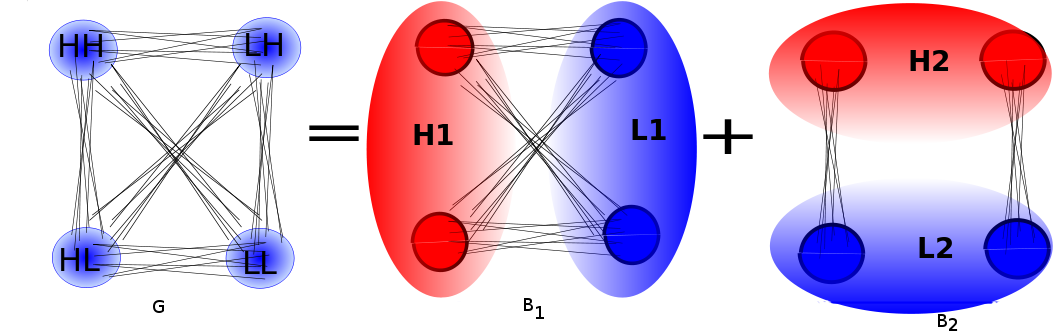}
\caption{{ Two dimensional decomposition of a graph.
}}
\label{fig:2D_example_new}
\end{figure}
The block diagram of a 
$2$ ``dimensional''
graphBior filterbank  is 
shown in Figure~\ref{fig:filterbank_imp2}, 
where a dimension is 
interpreted as filtering 
and downsampling on a single 
bipartite subgraph.
\begin{figure}[htb]
\centering
\includegraphics[width=5in]{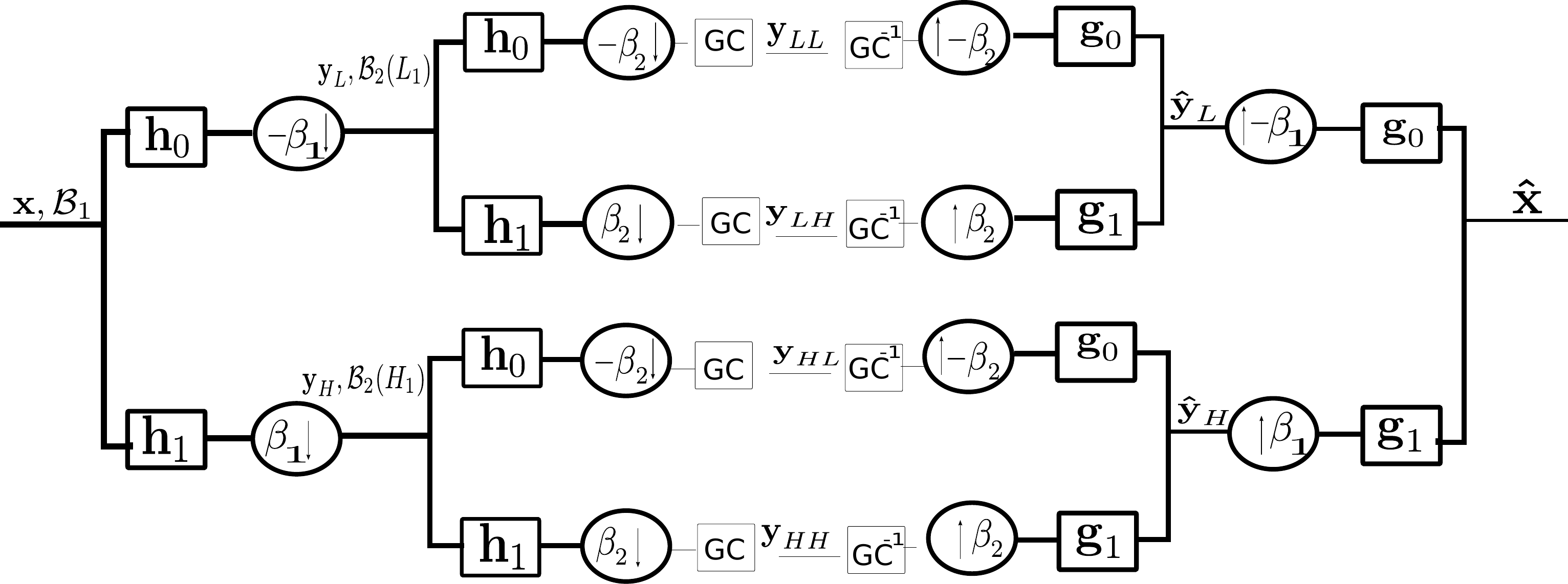}
\caption{ Separable two-dimensional filterbank on graphs.  
The graph is first decomposed into two bipartite subgraph as shown 
in Figure~\ref{fig:2D_example_new}. The binary function $\beta_1$ is such that 
$\beta_1(H_1) = 1$ and $\beta_1(L_1) = -1$. Similarly the binary function 
$\beta_2$ is such that $\beta_1(H_2) = 1$ and $\beta_1(L_2) = -1$.
For each bipartite graph the  graph transform pair $\{h_{0},h_1\}$ forms the 
analysis low-pass and analysis high-pass 
graphBior filters respectively and $\{g_{0},g_1\}$ 
are corresponding synthesis filters. GC: gain-compensation block, GC$^{-1}$: inverse GC block.
}
\label{fig:filterbank_imp2}
\end{figure}
Note that this design is analogous to separable filterbank 
implementation on regular multidimensional signals. 
For example in the case of separable
transforms for 2D signals, filtering in one dimension (e.g., row-wise) is followed by filtering of
the outputs along the second dimension (column-wise). Moreover, the separable graphBior 
filterbanks are PR for any arbitrary partitions $LL,LH,HL$ and $HH$ induced on the graph. 
The choice of a specific bipartite subgraph decomposition depends on 
various factors. For highly structured graphs 
such as graph representation of regular signals, the bipartite 
subgraphs which preserve the structure are more useful (see, for example, Section~\ref{sec:edge_aware_app}). 
For arbitrary graphs, there can be various criteria. One criterion is to compute 
a graph decomposition that generates minimum 
number of bipartite subgraphs
whose union covers all the links 
in the graph. 
An example of decomposition scheme based on such criterion 
is {\em Harary's algorithm} proposed in \cite{SunilTSP}, which 
provides a $\lceil log_2 K\rceil$ bipartite 
subgraph decomposition of a $K$-colorable graph\footnote{A $K$-colorable 
graph can be divided into $K$ clusters such that there are no links connected nodes in the same clusters. $\lceil~ . ~\rceil$ is the ceiling operator.}.
Another criterion introduced in~\cite{ICASSP12Sunil}, proposes subgraph decompositions 
so that the neighborhood sets of each node on 
different bipartite subgraphs are maximally disjoint. 
This leads to uncorrelated 
filtering operations on different graphs. 
However, whether the above mentioned decomposition schemes are optimal in some sense, or more generally 
whether there are other ways to extend graphBior filterbanks to arbitrary graphs, is part of 
our on-going research.


The multiresolution decomposition (MR) property in graphs implies successive coarser approximations 
of the graph and graph signal. For example, in a $1$-dimensional implementation, the output samples 
in the set $L$ are treated as signal for the next resolution level, and the vertices in $L$  
are reconnected to form a downsampled graph that
preserves properties of the original graph such as the
intrinsic geometric structure (e.g., some notion of distance
between vertices), connectivity, graph spectral distribution,
and sparsity. 
The graph coarsening problem has received a great deal
of attention from graph theorists, and, in particular, from the
numerical linear algebra community (see \cite{ron,gp_archive} and the reference therein). 
Further, Pesenson (e.g., \cite{pesenson_paley}) has leveraged the analogy
between the graph Fourier transform and the classical Fourier
transform to extend the concept of bandlimited sampling to
signals defined on graphs. Namely, certain classes of signals
can be downsampled on particular subgraphs and then stably
reconstructed from the reduced set of samples. In our designs, any of 
the above mentioned coarsening scheme can be used to compute the graph at the next level.

 \section{Experiments}
\label{sec:experiments}

\subsection{Performance comparisons of two-channel filterbanks on graphs}
In order to compare various graphBior designs proposed in this paper and previously proposed graphQMF designs, we simulate
$M$ instances of random graphs. In all the experiments the random graphs are bipartite graphs with 
$300$ nodes in each partition and probability of connection $2log(N)/N$. The isolated vertices in the graph 
are removed in each realization. 

In order to show the trade-off between vertex domain and spectral domain localizations,  we plot in 
Figure~\ref{fig:spatial_spectral_plots}, the spatial spread~(\ref{eq:spatial_spread}) 
and spectral spread~(\ref{eq:spectral_spread}) of various 
two-channel spectral filterbanks on $M=10$ instances of 
random bipartite graphs. We first observe that that the graph-QMF based on ideal 
half-band kernels (magenta diamonds in the plot)
have very small spectral spread but very large spatial spread, as compared to other designs. This is due to the brick-wall 
spectral response of these filterbanks. The same graphQMF filterbanks when designed using smooth Meyer kernel based
half-band filters (black squares in the plot), have lower spatial spread (though still higher than most of the graphBior filterbanks) 
but higher spectral spread. However, both of these designs do not have a compact support. 
On the other hand, the proposed graphBior filterbanks exploit the spatial/spectral tradeoff better and have compact support support.
The filters with 
smaller filterlengths are spatially more localized but spectrally less localized 
than the filters with higher filterlengths. The filter length of graphBior designs is chosen to be the maximum 
of the two filter lengths (i.e, $K$).  
Among graphBior designs, 
the zeroDC filterbanks (red triangles in the plot) perform slightly 
worse than the nonzeroDC filterbanks, which is due to the 
extra normalizations introduced in the formers to make 
their DC response zero. 
\begin{figure}[htb]
 \begin{center}
\subfigure[]{
   \includegraphics[width = 3in] {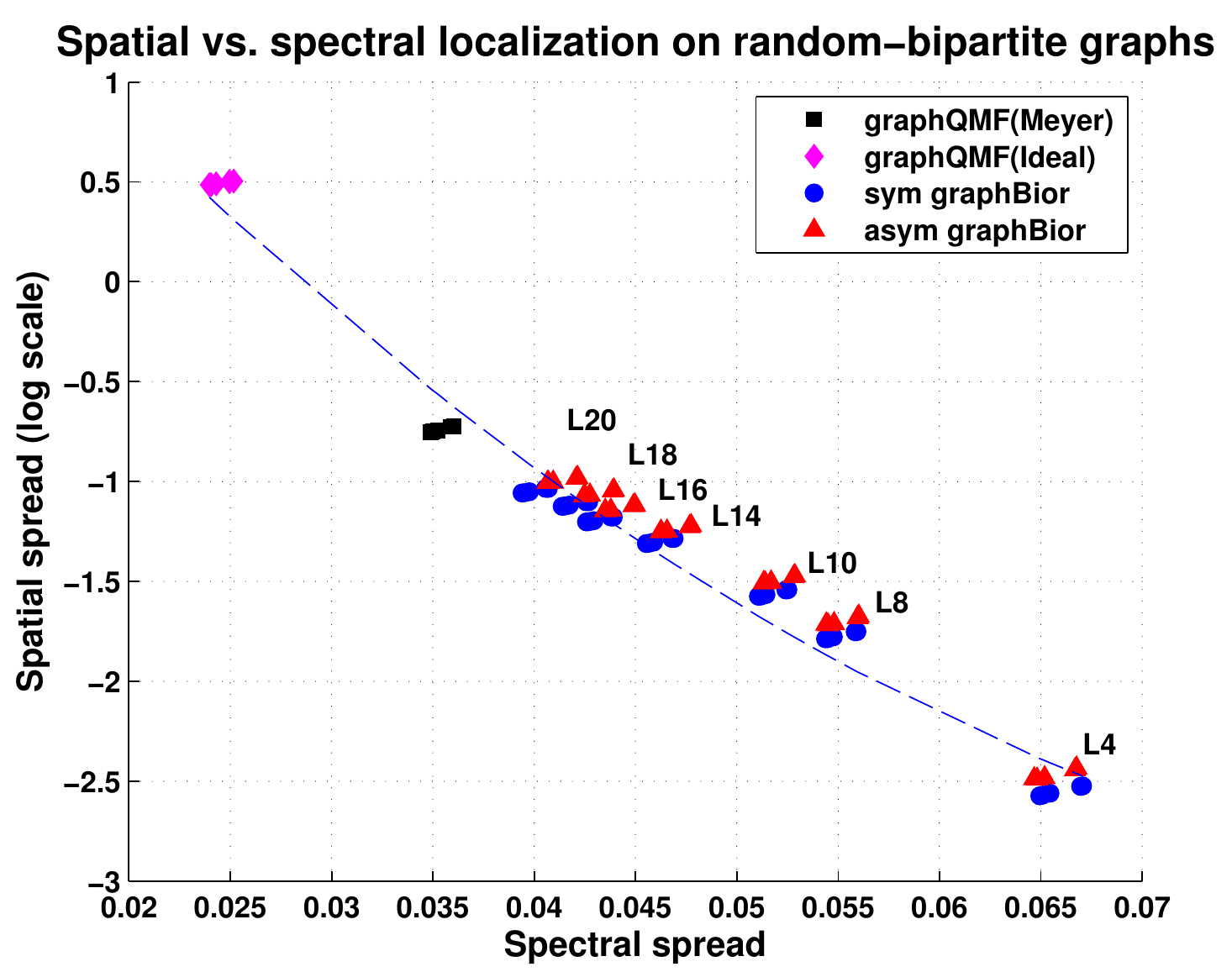}
   \label{fig:spatial_spread_HP}
 }
\subfigure[]{
   \includegraphics[width = 3in] {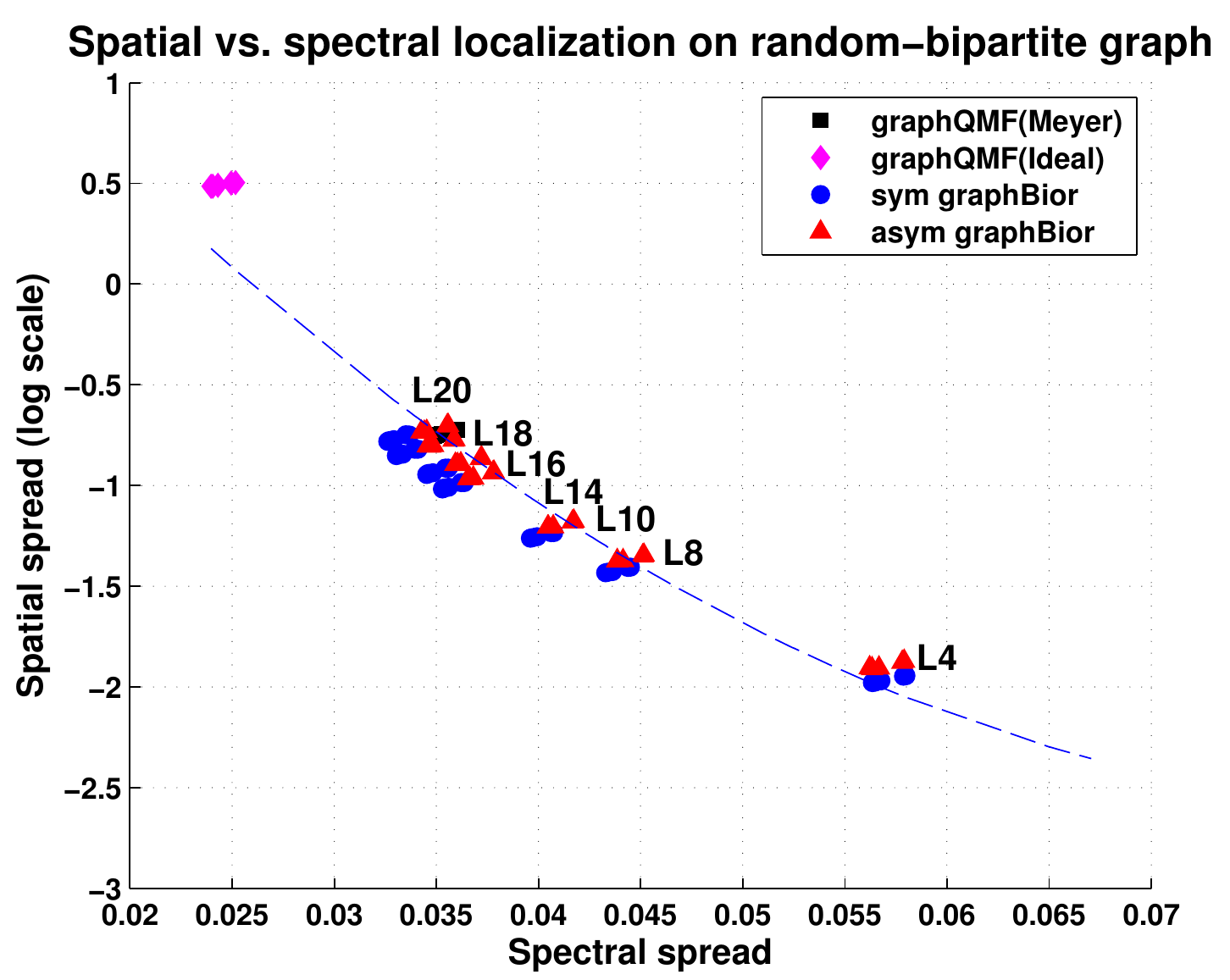}
   \label{fig:spatial_spread_LP}
 }
\caption{The spatial vs. spectral spread a) of the highpass filters, and b) of the lowpass filters. The spatial/spectral coordinates 
for graphQMF filterbanks are same in both plots, since the 
lowpass and highpass filters are symmetric around $\lambda = 1$. However, the lowpass and highpass filterbanks in graphBior 
designs are neither symmetric nor equal length. Therefore, the spatial/spectral spreads of the two channels are different.
The dashed line is a quadratic polynomial fit of the data points in the least square sense. 
}
\label{fig:spatial_spectral_plots}
\end{center}
\end{figure}


The exact graphQMF filterbanks provide PR but are not compact support. In~\cite{SunilTSP}, 
we 
proposed polynomial approximation of the exact graphQMF kernels 
which are compact support but results in some reconstruction error. 
A comparison between proposed graphBior filterbanks and the graph-QMF filterbanks, 
in terms of 
perfect reconstruction error (SNR) and orthogonality ($\Theta$) is shown 
in Table~\ref{tab:SNR_compare}. 
The reconstruction SNR and orthogonality 
$\Theta$ are computed as an average over $20$ instances of 
randomly generated graph-signals on $M = 10$ random  bipartite 
graphs. 
It can be seen from Table~\ref{tab:SNR_compare} 
that all graphBior designs provide perfect reconstruction ($SNR > 100dB$). 
The graph-QMF filters 
in comparison are closer to orthogonal (i.e., $\Theta$ almost $1$), 
but have considerably lower reconstruction SNR. We now consider 
some applications of our proposed filterbanks.
\begin{table}[htbp]
\begin{center}
\begin{tabular}{|p{1 cm}|p{1 cm}|p{1 cm}|p{1 cm}|p{1 cm}|p{1 cm}|p{1 cm}|}
\hline
\multicolumn{1}{|p{1 cm}|}{\textbf{L}} & \multicolumn{1}{r}{\textbf{Graph}} & \multicolumn{1}{l|}{\bf QMF} & \multicolumn{1}{r}{\textbf{nonzeroDC}} & \multicolumn{1}{l|}{\bf graphBior} & \multicolumn{1}{p{1 cm}}{\textbf{zeroDC }} & \multicolumn{1}{l|}{\bf graphBior} \\ \hline
\multicolumn{1}{|l|}{} & \multicolumn{1}{l|}{\textbf{SNR (dB)}} & \multicolumn{1}{l|}{\textbf{$\Theta$}}  & \multicolumn{1}{l|}{\textbf{SNR (dB)}} & \multicolumn{1}{l|}{\textbf{$\Theta$}}  & \multicolumn{1}{l|}{\textbf{SNR (dB)}} & \multicolumn{1}{l|}{\textbf{$\Theta$}} \\ \hline
4 & 32.20 & 0.98 &  286.84 & 0.88 & 286.54 & 0.70  \\ \hline
8 & 32.25 & 0.98 & 282.89 & 0.87 &  282.71 & 0.66  \\ \hline
10 & 42.17 & 1.00&  270.05 & 0.81 & 270.00 & 0.65  \\ \hline
14 & 48.09 & 1.00 & 230.83 &0.85 &  230.73 & 0.64 \\ \hline
16 & 44.78 & 0.99 & 222.08 & 0.94 &  222.05 & 0.64  \\ \hline
18 & 45.23 & 0.99 & 190.53 & 0.92 &  190.43 & 0.63  \\ \hline
20 & 54.61 & 1.00 & 170.78 & 0.94 &  170.68 & 0.63  \\ \hline
\end{tabular}
\end{center}
\caption{Comparison between graph-QMF filterbanks (polynomial approximations) and graphBior filterbanks on random bipartite graphs.}
\label{tab:SNR_compare}
\end{table}

\subsection{Graph based image processing}
\label{sec:edge_aware_app}
In this section, we describe an
application of proposed  
graphBior filterbanks 
for image-analysis.
This is an extension of our 
previous work in~\cite{SunilTSP,SSP'12}, where 
we proposed a graph 
based edge-aware representation 
of image-signals. 
While standard separable extensions
of wavelet filterbanks to higher 
dimensional signals, such as $2$-D images, 
provide useful multi-resolution analysis, 
they do not capture the intrinsic 
geometry of the images. 
For example, these extensions  
can capture only limited 
(mostly horizontal and vertical
) directional information. Images can also be viewed as graphs, by treating pixels as nodes, pixel intensities as graph-signals, 
and by connecting pixels with their neighbors in various ways. The advantage of formulating images 
as graphs is that different graphs 
can represent the same image, 
which offers flexibility of 
choosing the graphs
that have 
useful properties.
In~\cite{SunilTSP}, we proposed an 
$8$-connected graph representation 
of images, in which each pixel is connected to $8$ of its nearest 
neighbors ($4$ diagonal, $2$ vertical and $2$ horizontal)
as shown in Figure~\ref{fig:2D_decomp}.
\begin{figure}[htb]
\centering
\includegraphics[width=5in]{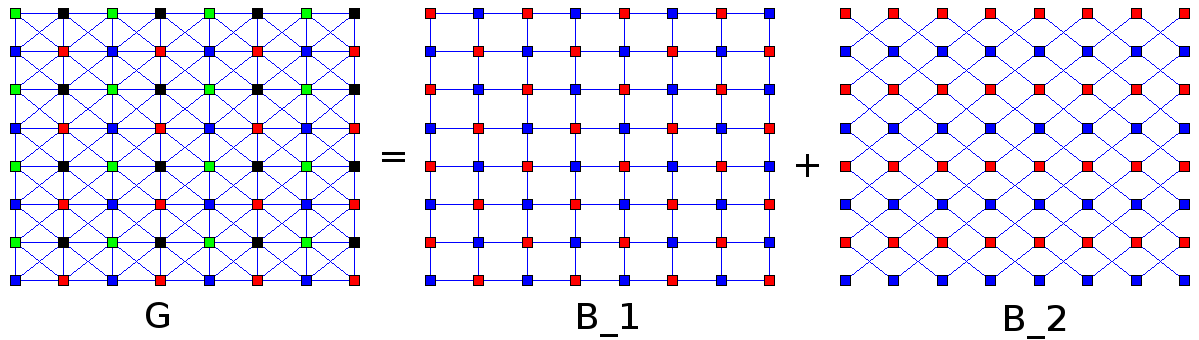}
\caption{{ Two dimensional decomposition of $8$-connected image-graph}}
\label{fig:2D_decomp}
\end{figure}
The graph is not bipartite, but can be decomposed into two bipartite 
subgraphs, one containing links in the horizontal 
and vertical direction and 
other in the diagonal directions. 
The proposed 
graphBior filterbanks can then be applied 
as two ``dimensional''  
filterbanks as in Figure~\ref{fig:filterbank_imp2}. The advantage 
of using graphBior wavelet filterbanks as against standard separable filterbanks, 
is that the former provide more filtering directions (diagonal and rectangular)
than the latter (only rectangular), at the same order of 
computational complexity.  
For multi-resolution analysis, 
the downsampled set of nodes in the $LL$ channel, are again 
connected to $8$ of their neighbors, 
to create a downsampled graph
and the graphBior filterbanks are 
implemented iteratively on the 
downsampled graphs. Thus the downsampling ratio at each level 
is same in both graphBior filterbanks and standard separable filterbanks. 

In~\cite{SSP'12}, we proposed an {\em edge-aware} 
implementation for piece-wise smooth images,  
in which the bipartite subgraphs 
obtained in 
Figure~\ref{fig:2D_decomp},
can be simplified 
by removing the 
links between pixels 
%
across 
which the pixel intensity 
changes drastically. 
These links can be found 
using any standard 
edge-detection algorithm 
(we use Canny edge detection in our 
experiments and remove connected 
components
less than $50$ pixels 
before computing the graph). 
The advantage 
of edge-aware graph representations 
is 
that it avoids filtering 
across edges, 
which leads 
to a very
significant reduction 
in the number of 
large coefficients 
near edge
(and thus corresponding 
reductions in rate). Note that in a compression application, this would
require generating an edge map at the 
encoder and then sending it to
the decoder. However, 
recent work~\cite{shen2010edge,wooshik} using 
transforms 
based on similar edge-map information have shown that 
even with the
extra overhead of sending the edge map we can achieve reductions
in overall transmitted rate.
%
%
%

In order to demonstrate the 
advantage of graph based implementation 
of proposed filterbanks, we choose {\em coins.png}
image as shown in Figure~\ref{fig:coins} with many round 
shaped coins. We implement graphBior 
filterbanks of length $10$, (i.e., $graphBior(5,5)$), 
and compare them 
against standard separable CDF $9/7$ 
wavelet filterbanks in a non-linear 
approximation of images, using $4$ 
resolution levels. 
Figure~\ref{fig:reconstruction_coins} shows the 
reconstruction of coins.png using all lowpass coefficients and 
top $4$\% of the highpass coefficients in terms of magnitude. 
\begin{figure}[htb]
 \begin{center}
\subfigure[]{
   \includegraphics[width = 2in] {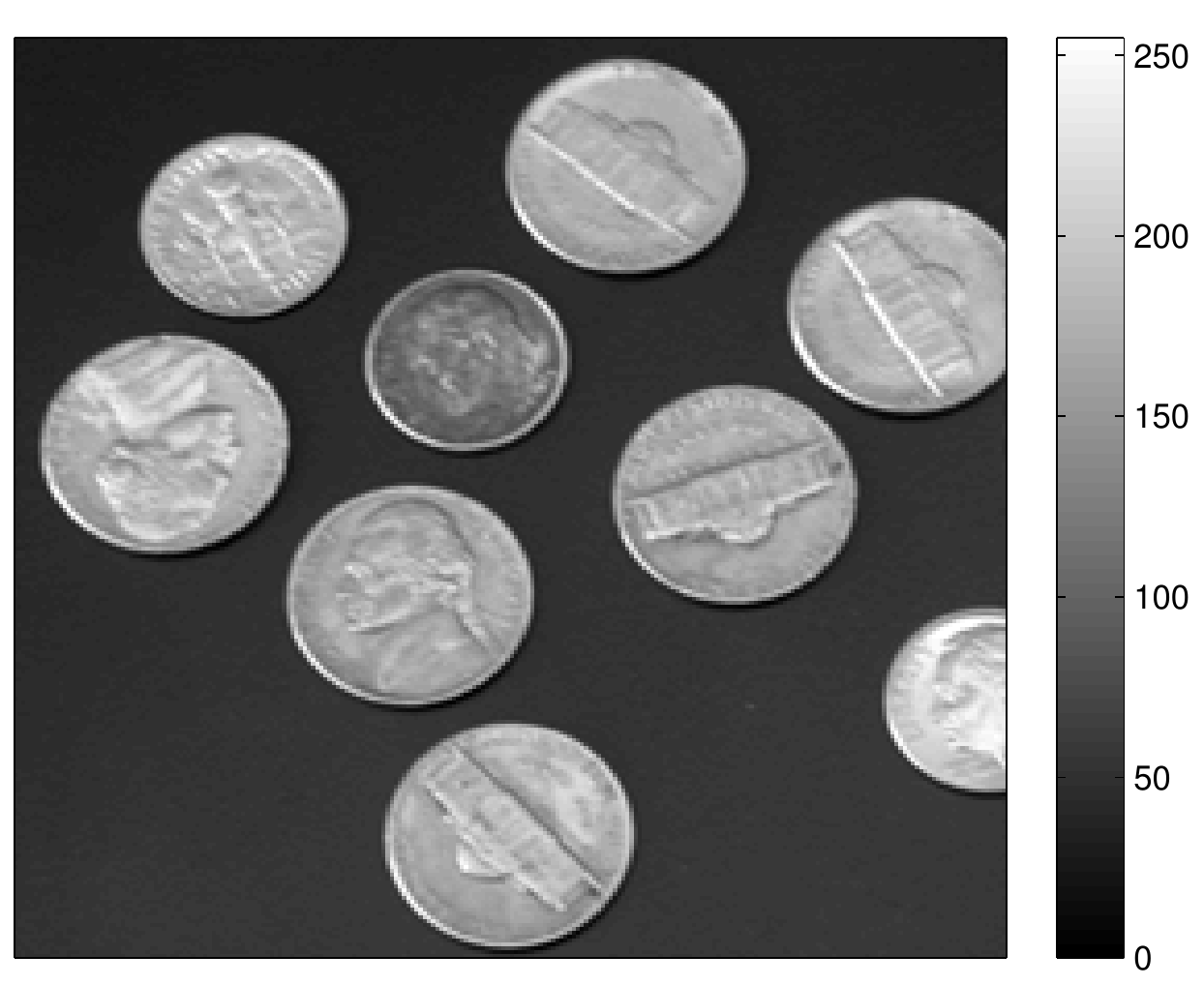}
   \label{fig:coins}
 }
\subfigure[]{
   \includegraphics[width = 2in] {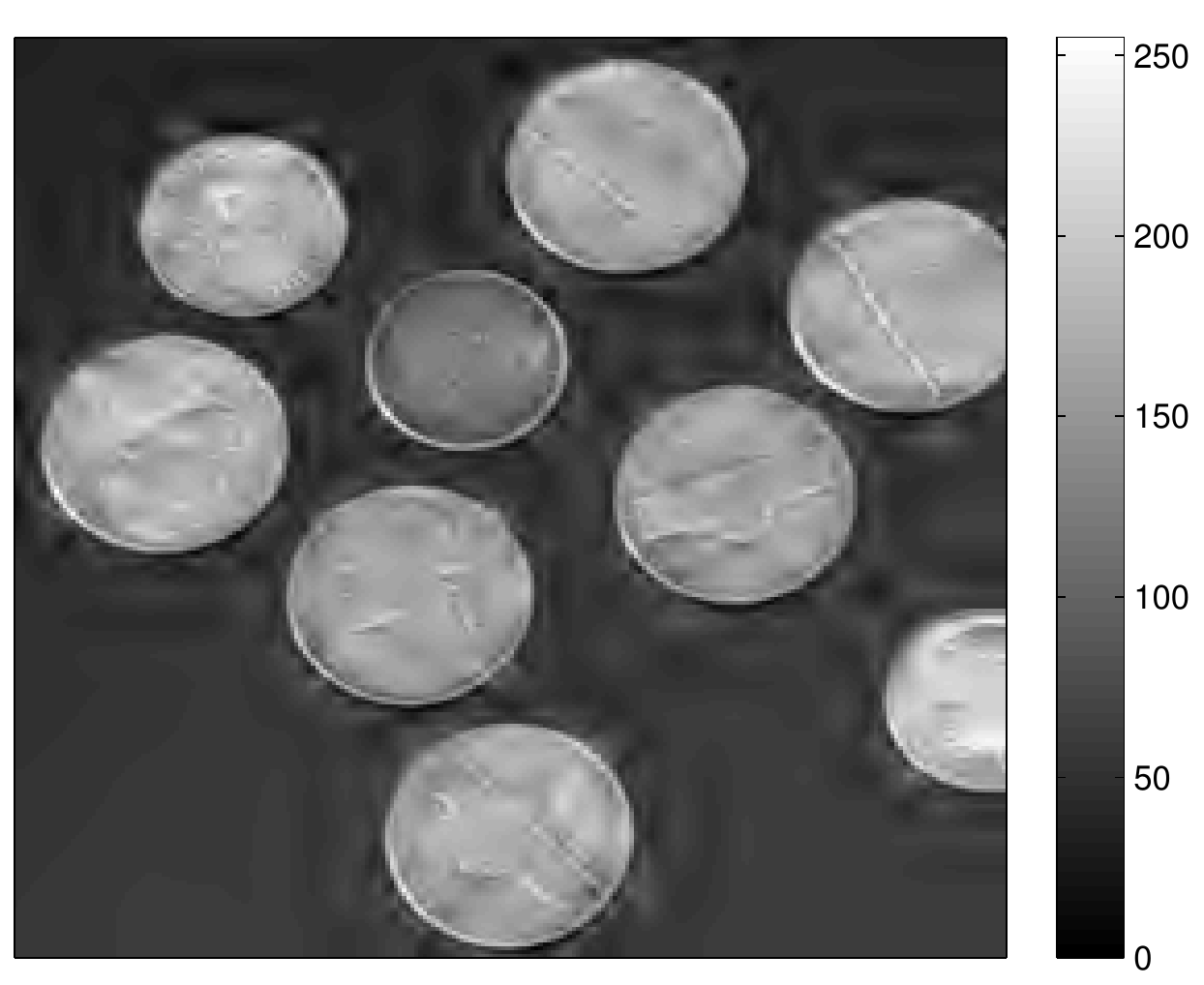}
   \label{fig:CDF}
 }
 \subfigure[]{
   \includegraphics[width = 2in] {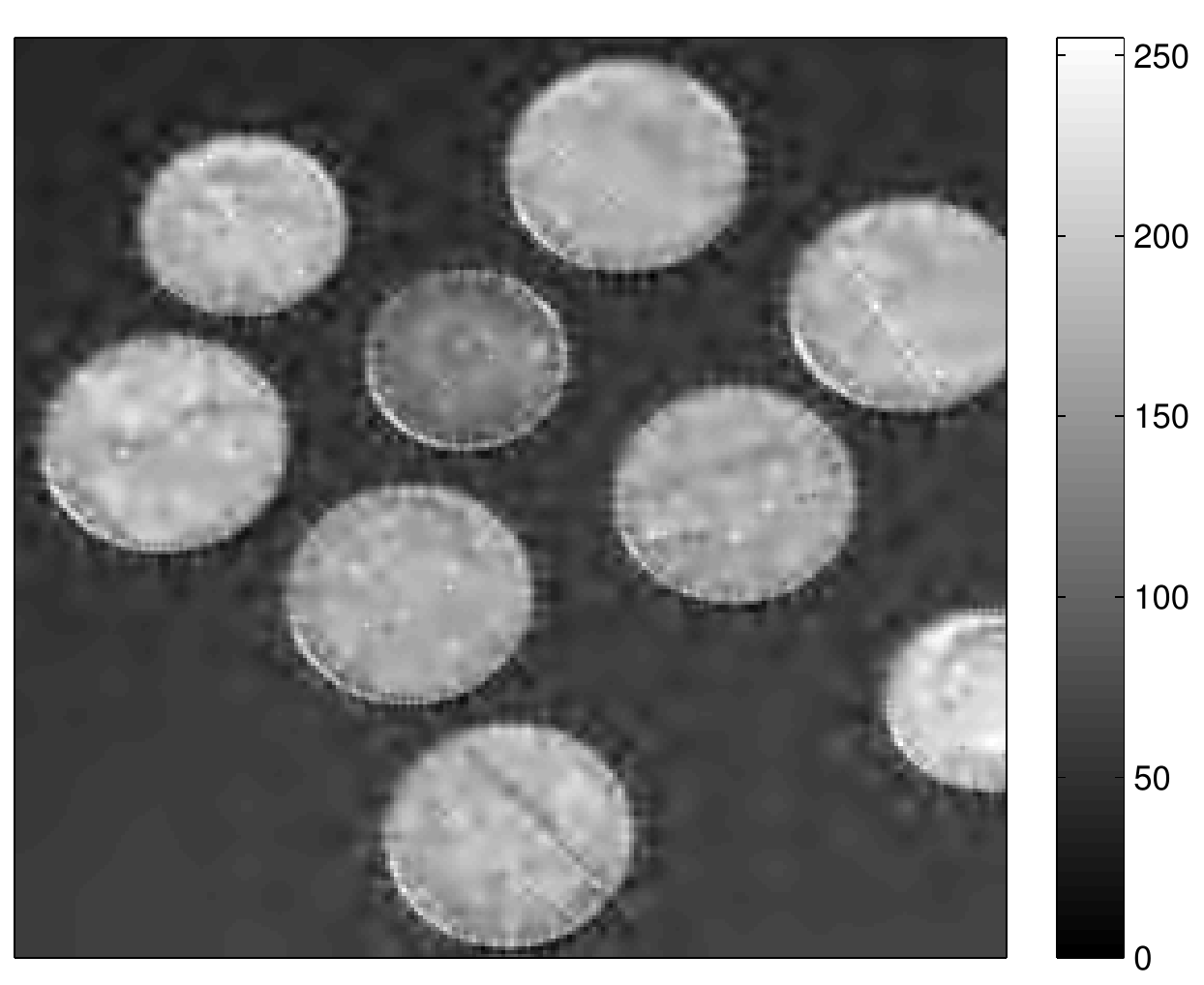}
   \label{fig:zeroDC graphBior}
 }
 \subfigure[]{
   \includegraphics[width = 2in] {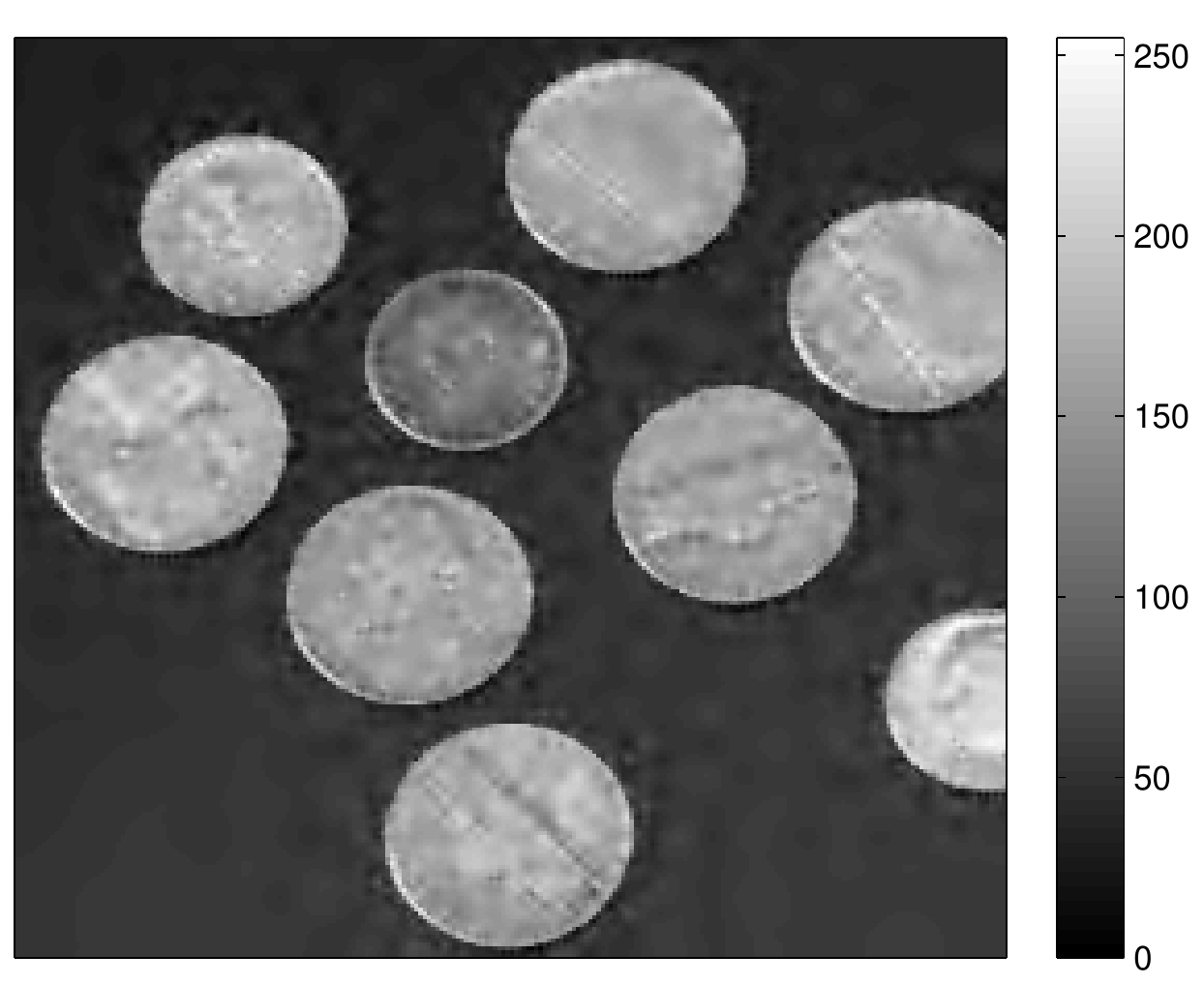}
   \label{fig:edge-aware graphBior}
 }
  \subfigure[]{
   \includegraphics[width = 2in] {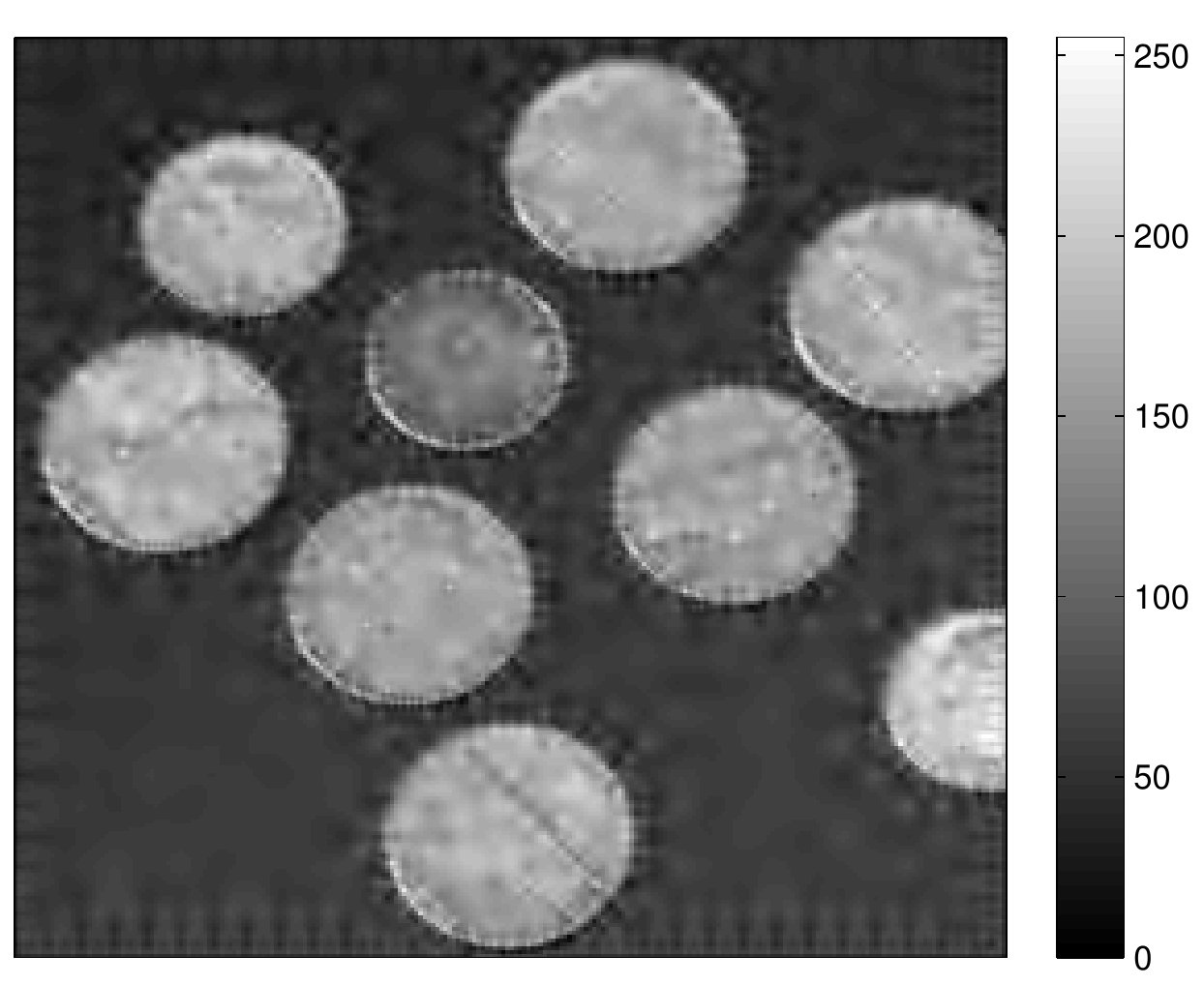}
   \label{fig:nonzeroDC graphBior}
 }
 \subfigure[]{
   \includegraphics[width = 2in] {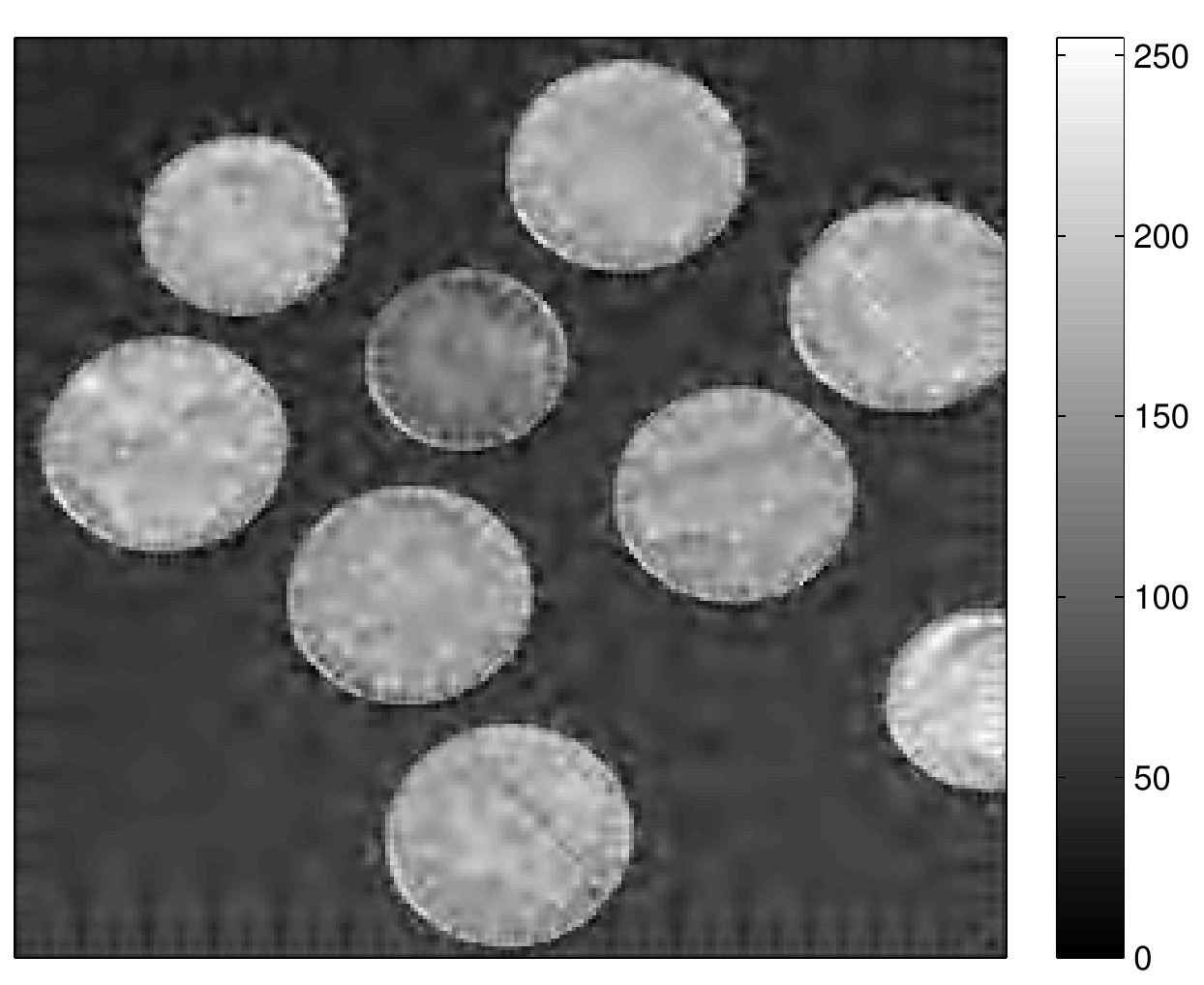}
   \label{fig:edge-aware nonzeroDC graphBior}
 }
\caption{ Reconstruction of ``Coin.png'' ($512 \times 512$) from all lowpass coefficients and $3\%$ highpass coefficients after a $4$-level decomposition. (a) Original image, 
(b) standard CDF $9/7$ filters, (c) zeroDC filterbanks on regular $8$-connected image graph 
(d) zeroDC filterbanks on edge-aware image graph (e) nonzeroDC filterbanks on regular $8$ connected image graph and (f) nonzeroDC filterbanks on edge-aware image graph.}
\label{fig:reconstruction_coins}
\end{center}
\end{figure}
Since the standard 
separable wavelet filterbanks 
filter only in horizontal 
and vertical directions, they produce 
lots of 
large magnitude wavelet 
coefficients (and hence blurring artifacts) 
near the edges 
(see Figure~\ref{fig:CDF}). The 
zeroDC graphBior filterbank implementation 
on the regular $8$-connected 
image graph (Figure~\ref{fig:zeroDC graphBior}) 
does slightly better
since it also provide 
filtering in diagonal directions. 
%
However, the best performance in terms of 
reconstruction quality is observed 
for proposed 
edge-aware zeroDC graphBior 
filterbanks,
especially 
in preserving the edge structure
(see Figure~\ref{fig:edge-aware graphBior}).
This is due to the fact that the underlying graphs in 
this approach are disconnected at the edges, and hence the
filtering operations do not cross the edges. 
Theoretically, the nonzeroDC 
filterbanks should perform 
almost the same as zeroDC filterbanks 
for regular degree graphs. The $8$-connected image-graphs 
are almost regular except at the boundaries, and edges, and 
we observe in Figures~\ref{fig:nonzeroDC graphBior} 
and~\ref{fig:edge-aware nonzeroDC graphBior}, that 
significant ringing artifacts are produced near 
these places, when using 
nonzeroDC filterbanks. The problem of boundary 
artifacts also arises when using standard filterbanks 
on images, which is usally solved by providing 
signal extensions at the boundaries. Whether such signal 
extensions can be proposed for graph 
representation of images, 
is an open issue. 
%
%
%
%
%
%
Figure~\ref{fig:metric_compare}, 
shows PSNR and SSIM~\cite{SSIM} values plotted against 
fraction of detail 
coefficients used in the reconstruction of {\em coins.png}
image, and it can be seen from both the plots that zeroDC graphBior 
filterbanks perform better (up to 2dB better in PSNR) than 
the standard CDF $9/7$ filterbanks. Thus, the results show that the proposed 
graphBior filterbanks provide advantages over the standard wavelet transforms, with the same order
of computational complexity.

\begin{figure}[htb]
 \begin{center}
\subfigure[]{
   \includegraphics[width = 3.4in] {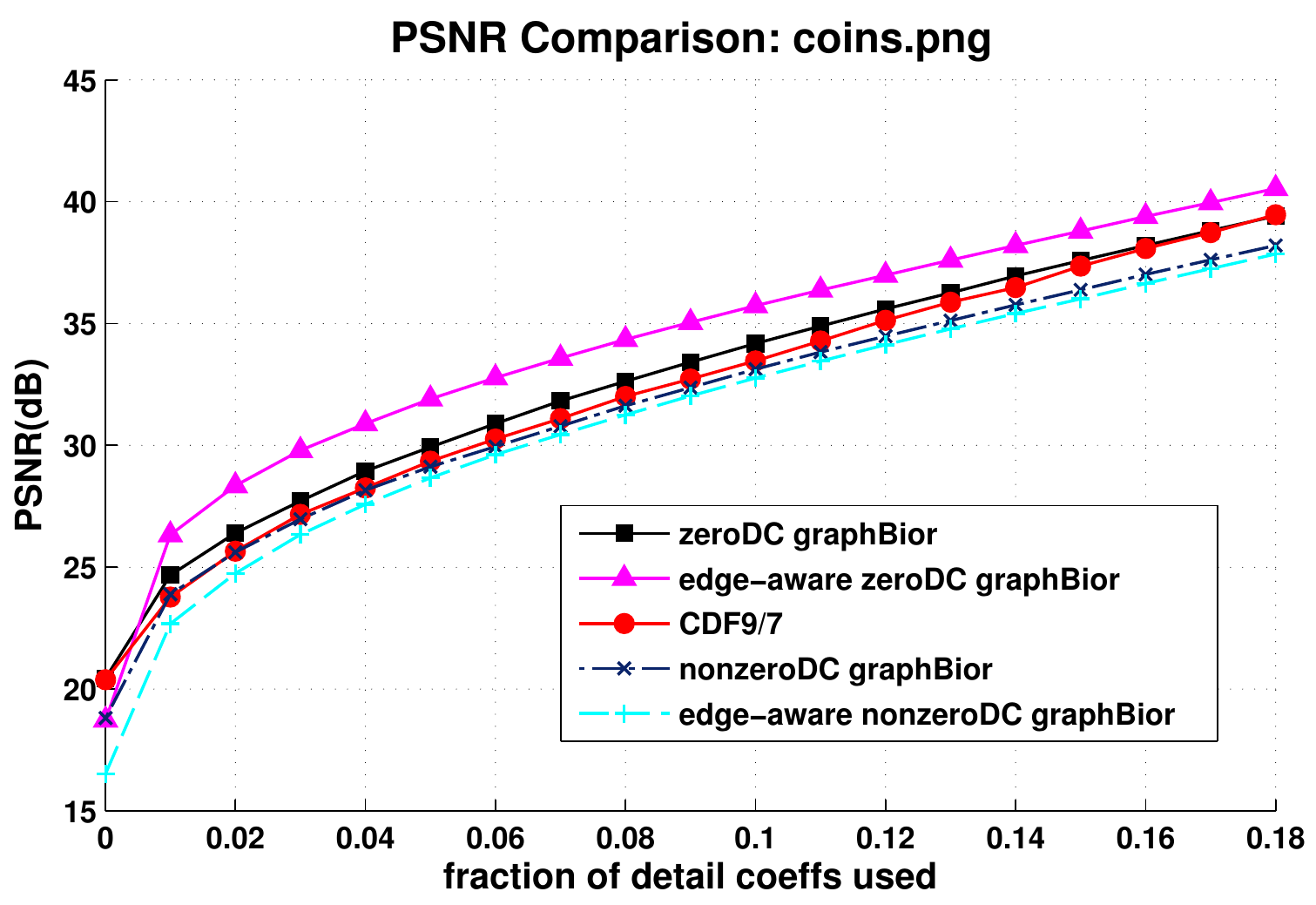}
   \label{fig:PSNR_compare}
 }
\subfigure[]{
   \includegraphics[width = 3.4in] {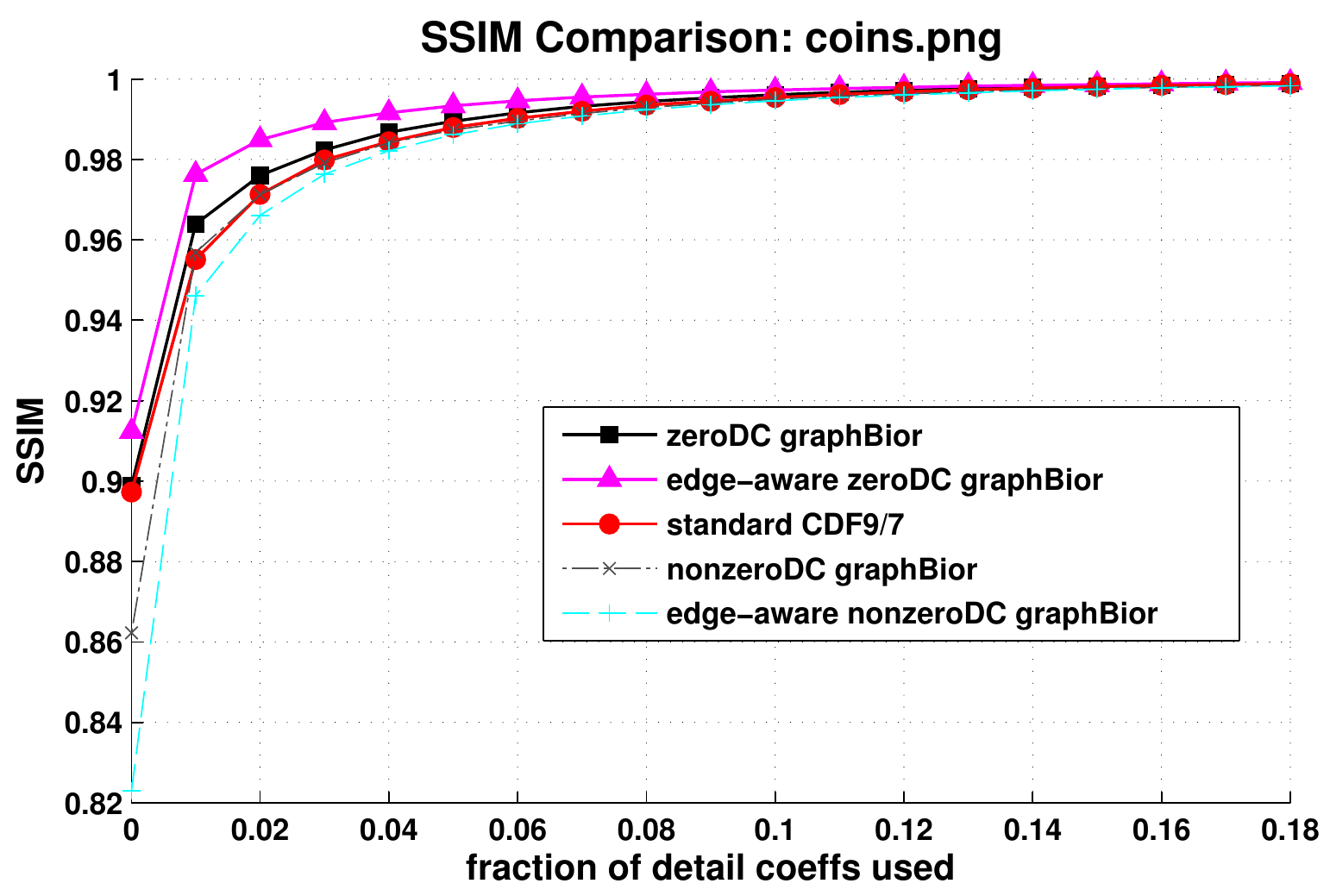}
   \label{fig:SSIM_compare}
 }
\caption{Reconstruction of coins.png image from all low-pass coefficient and a fraction of wavelet coefficients 
(sorted in the order of magnitudes). The fraction value is plotted on the x-axis. (a) PSNR of the reconstructed images,
(b) SSIM of the reconstructed image.
}
\label{fig:metric_compare}
\end{center}
\end{figure}

\subsection{Compression and Learning on arbitrary graphs}
The proposed filterbanks are 
useful in analyzing and compressing 
signals defined on arbitrary graphs. 
As a proof of concept,
we implement proposed graphBior 
filterbanks on the 
Minnesota traffic graph 
used in~\cite{SunilTSP}.
The graph is shown in 
Figure~\ref{fig:minnesota_inp}(a), and the graph signal to be analyzed is 
shown in 
Figure~\ref{fig:minnesota_inp}(b), where the color of a node 
represents the signal value at that node. The graph
is perfectly $3$-colorable and hence, it can be decomposed using Harary's decomposition~\cite{SunilTSP}
into $\lceil log_2(3) \rceil = 2$ bipartite subgraphs,
which are shown in 
Figure~\ref{fig:minnesota_inp}(c-d),
and a $2$-dimensional graphBior 
filterbank given in~Figure \ref{fig:filterbank_imp2} with filterlength $= 10$
is implemented on the graph.  
\begin{figure}[htb]
\begin{center}
\subfigure[]{
   \includegraphics[width = 2.1in] {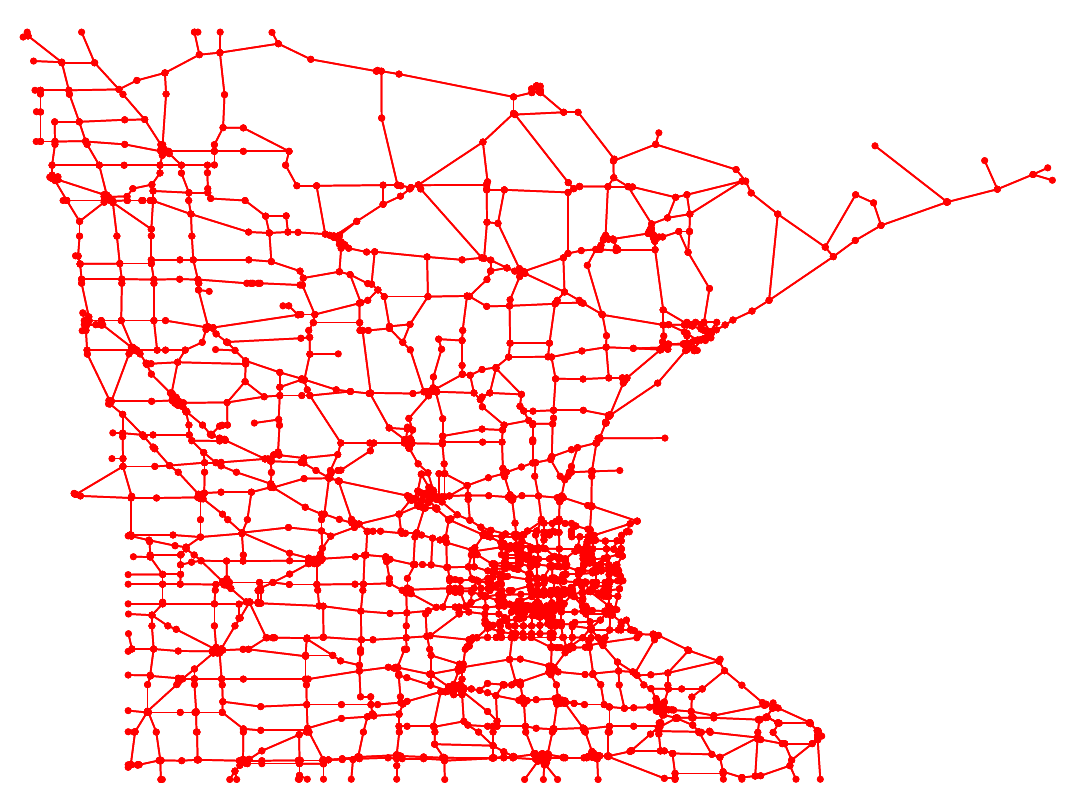}
   \label{fig:graph}
 }
\subfigure[]{
   \includegraphics[width = 2.1in] {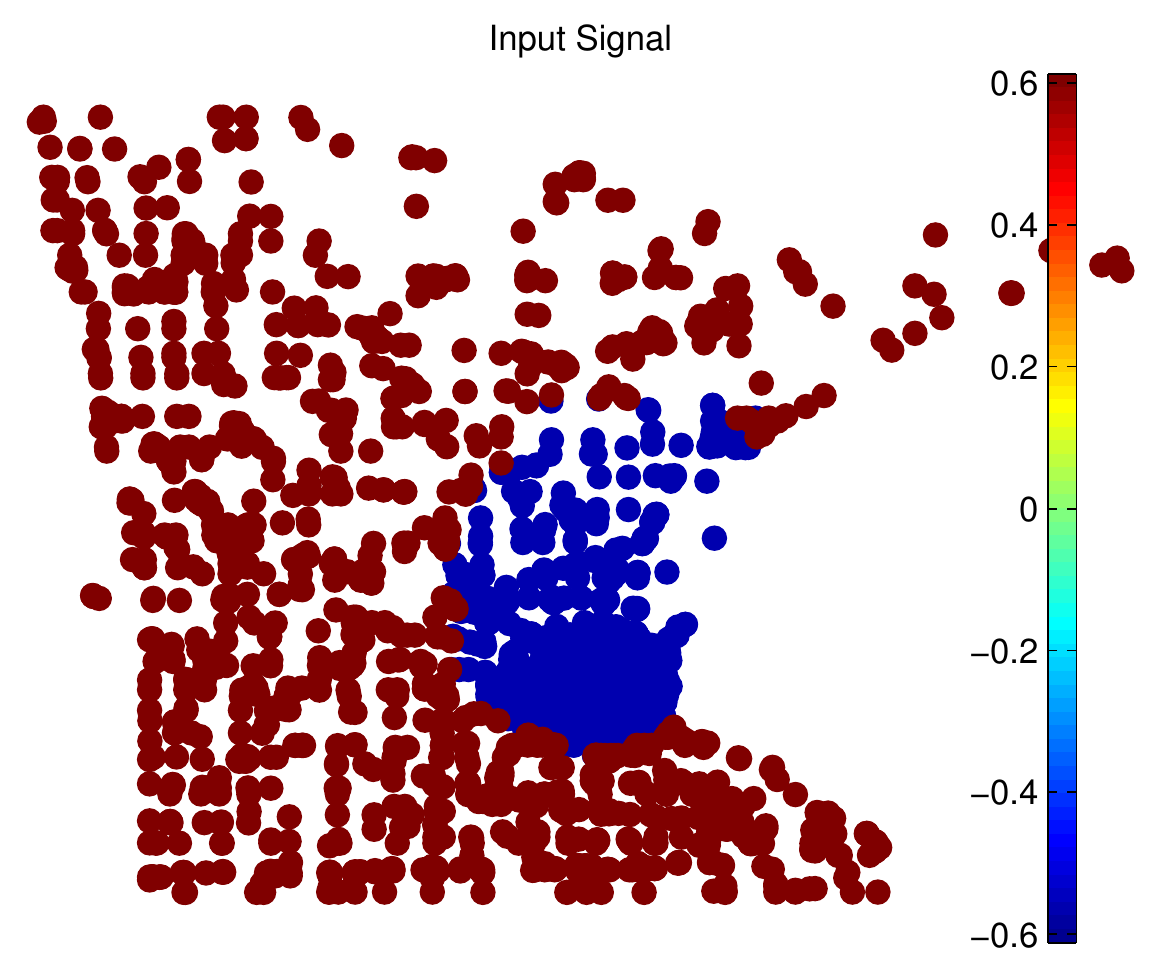}
   \label{fig:signal}
 }\\
\subfigure[]{
   \includegraphics[width = 2.1in] {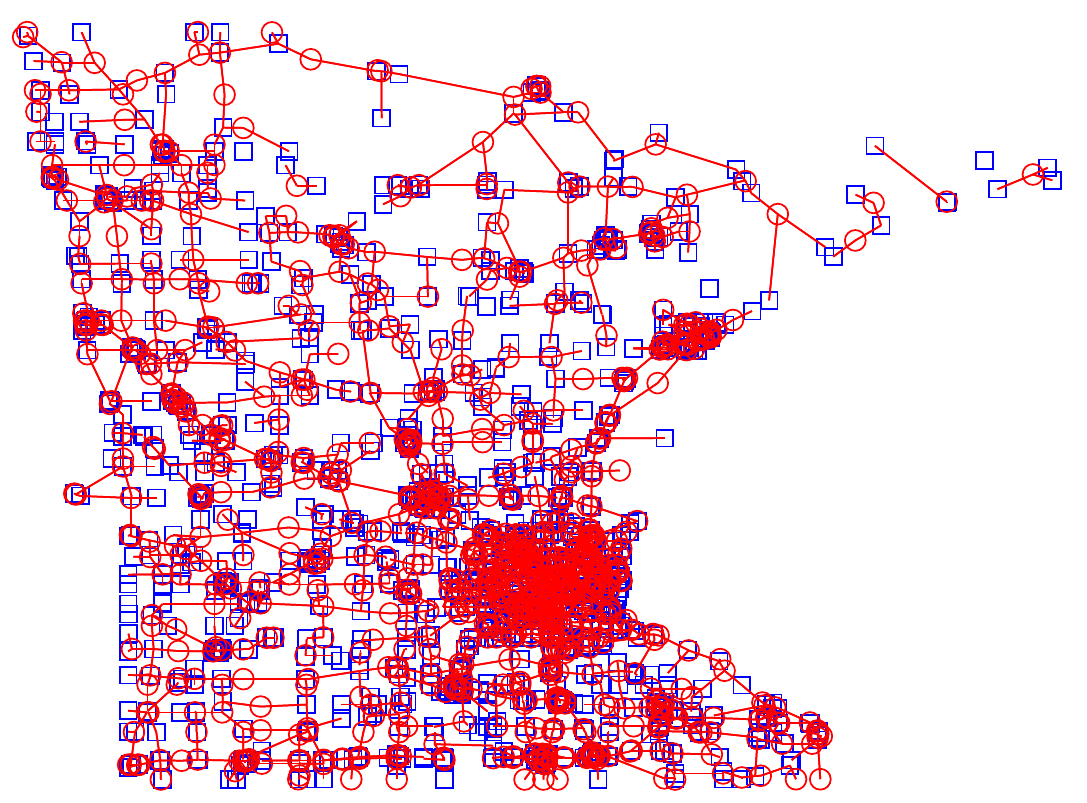}
   \label{fig:bpt1}
 }
\subfigure[]{
   \includegraphics[width = 2.1in] {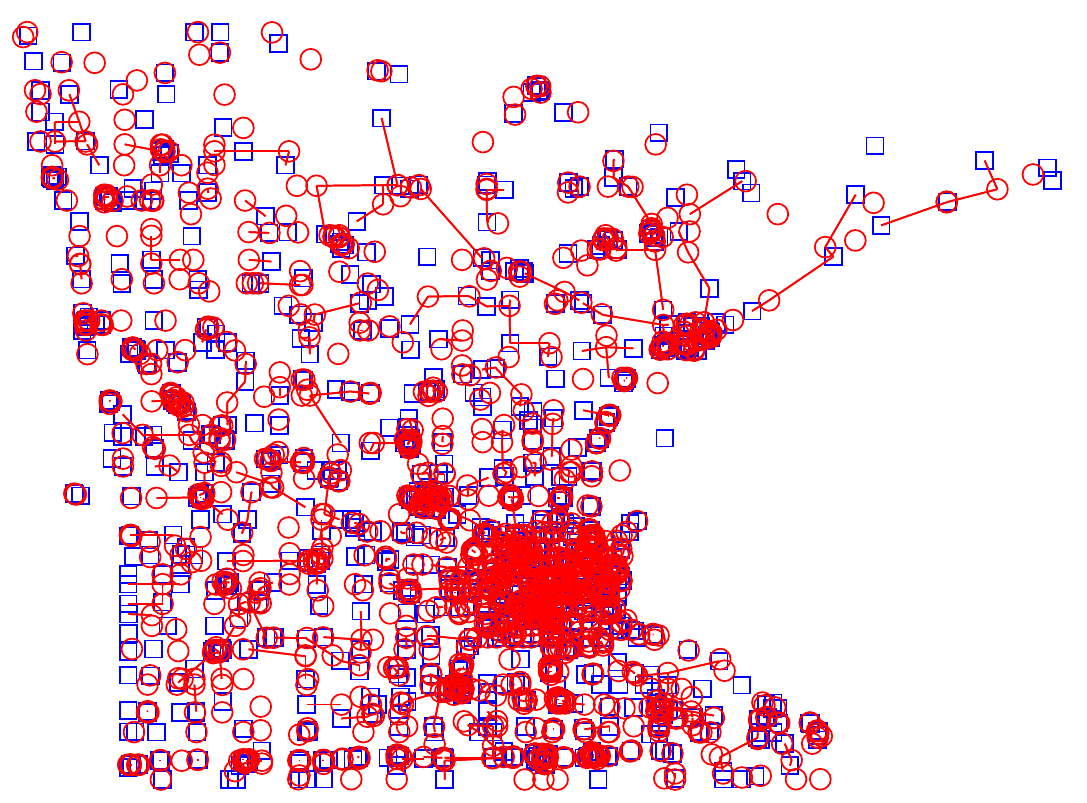}
   \label{fig:bpt2}
 }
\caption{(a) The Minnesota traffic graph $G$, and (b)  the
graph-signal to be analyzed. The colors of the nodes represent the sample values. 
(c)(d) bipartite decomposition of $G$ into two bipartite subgraphs using Harary's decomposition.}
\label{fig:minnesota_inp}
\end{center}
\end{figure}

The output in the $4$ channels can be interpreted as follows: 
the $LL$ channel providing a smooth 
approximation of the original signal on a subset of nodes, 
and the remaining channels providing details required 
for perfect reconstruction. 
Moreover, the total number of outputs in 
all channels is equal to the total 
number of input samples, hence 
the transform is critically sampled. The $HL$ channel does not sample any 
output and is empty. 
The output coefficients of both zeroDC and nonzeroDC 
filterbanks in $LL,LH$ and $HH$ channels 
are shown in Figure~\ref{fig:minnesota_wav}. 
Note that 
the graph-signal is 
piece-wise constant, 
hence the 
proposed zeroDC filterbanks (bottom row in Figure~\ref{fig:minnesota_wav})
provide a 
sparser approximation than the 
nonzeroDC filterbanks (top row in Figure~\ref{fig:minnesota_wav}).
\begin{figure}[htb]
\begin{center}
 \subfigure[]{
   \includegraphics[width = 2in] {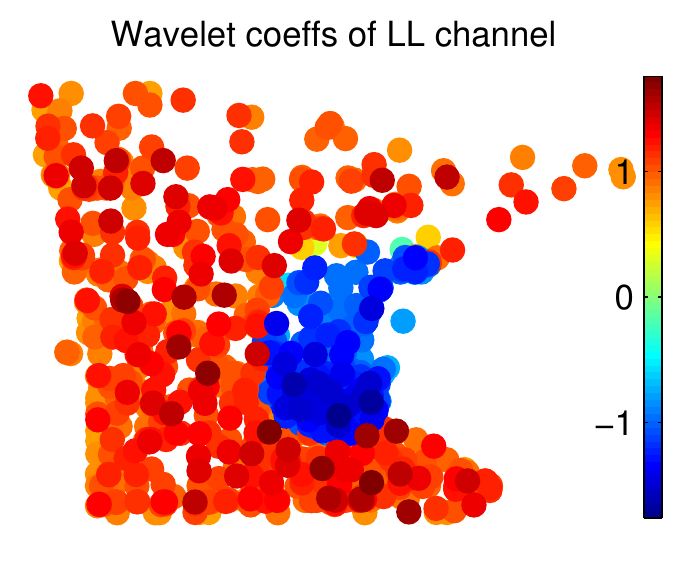}
   \label{fig:LL_sym}
 }
\subfigure[]{
   \includegraphics[width = 2in] {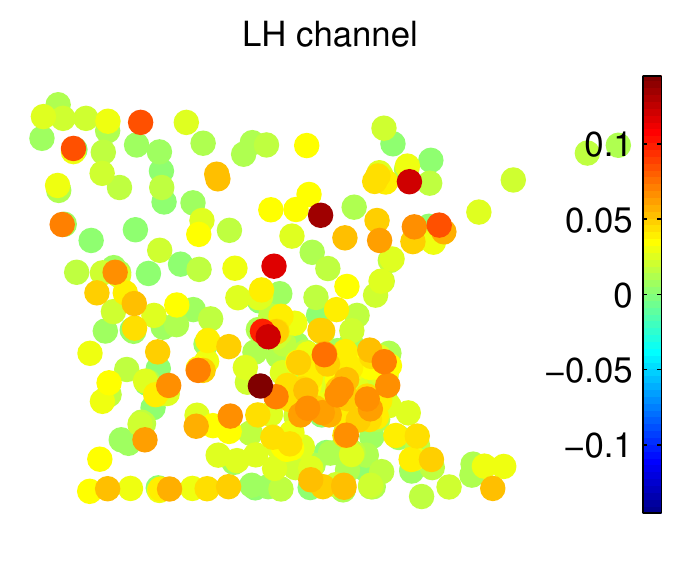}
   \label{fig:LH_sym}
 }
\subfigure[]{
   \includegraphics[width = 2in] {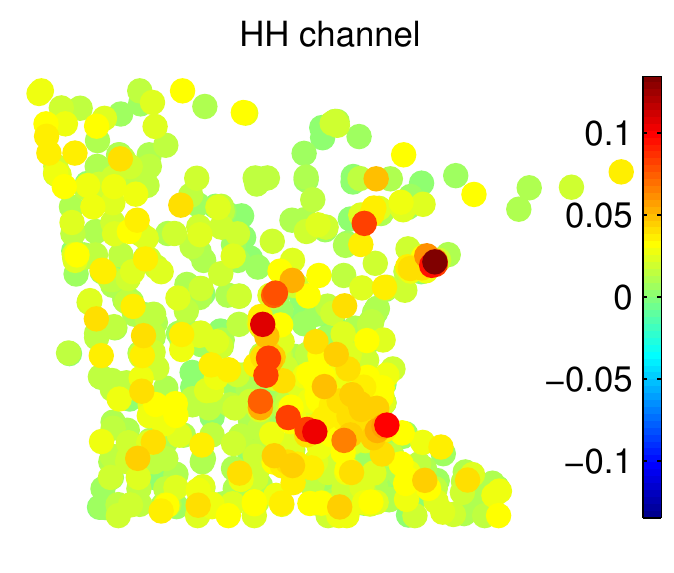}
   \label{fig:HH_sym}
 }\\
\subfigure[]{
   \includegraphics[width = 2in] {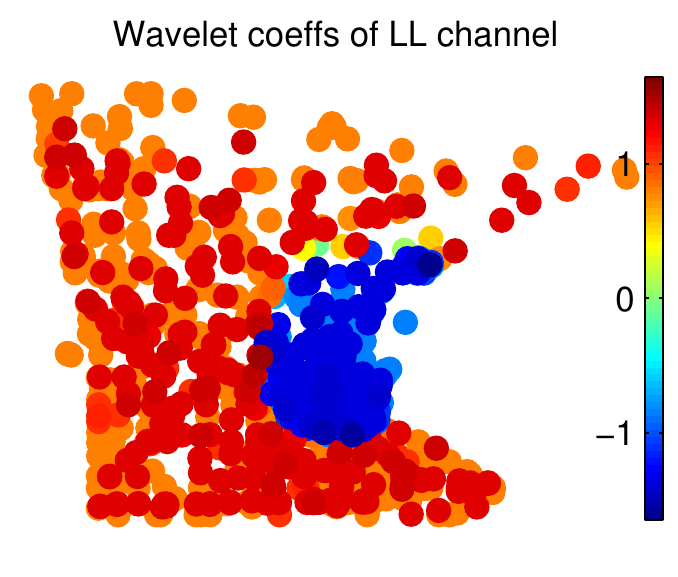}
   \label{fig:LL_asym}
 }
\subfigure[]{
   \includegraphics[width = 2in] {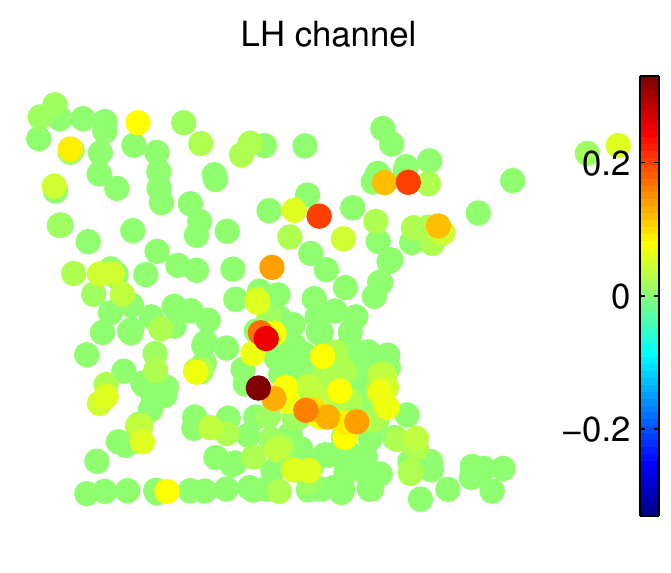}
   \label{fig:LH_asym}
 }
\subfigure[]{
   \includegraphics[width = 2in] {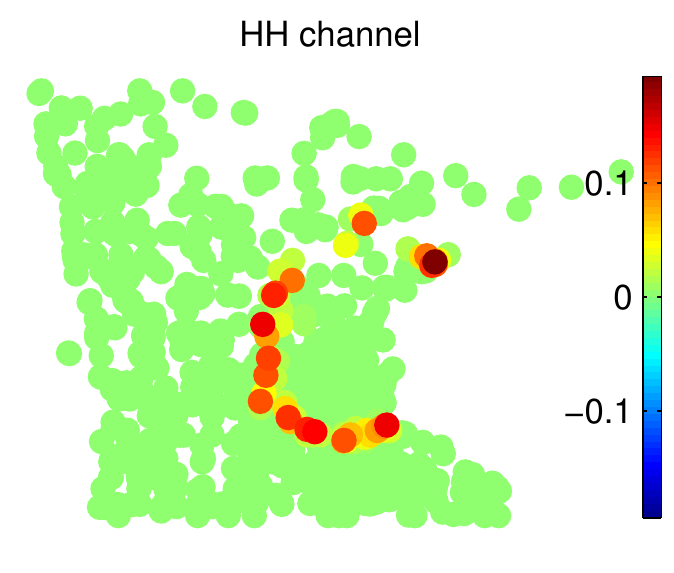}
   \label{fig:HH_asym}
 }

\caption{output coefficients 
of the graphBior filterbanks with parameter $(k_0,k_1) = (7,7)$. 
The node-color reflects the value of the coefficients at that point.  
Top-row: wavelet coefficients of nonzeroDC graphBior, 
bottom-row: wavelet coefficients of zeroDC graphBior, }
\label{fig:minnesota_wav}
\end{center}
\end{figure}
%
As a result, the non-linear approximation of the graph signal with only 
$1\%$ highpass coefficients (and all low pass coefficients) provide better SNR 
when using zeroDC filterbanks than when using nonzeroDC filterbanks as shown in 
Figure~\ref{fig:minnesota_outp}. 
\begin{figure}[htb]
\begin{center}
\subfigure[]{
   \includegraphics[width = 2.5in] {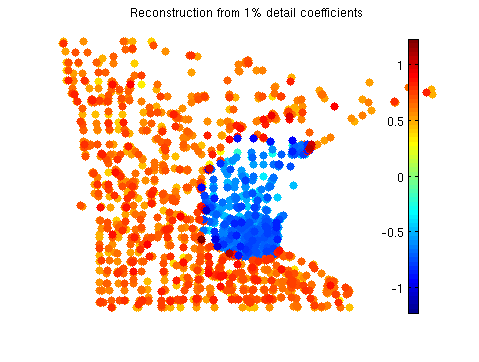}
   \label{fig:nonorm_reconst}
 }
\subfigure[]{
   \includegraphics[width = 2.5in] {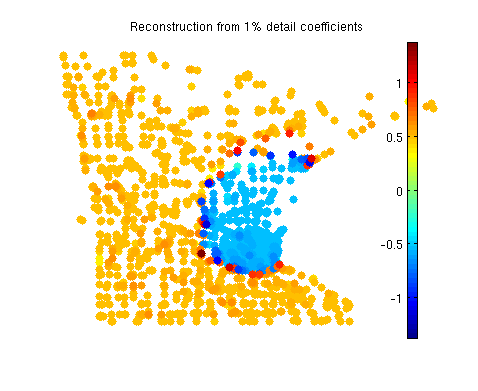}
   \label{fig:deg_norm_reconst}
 }
\caption{ Reconstructed graph-signals from  all coefficients of $LL$ channel and top $1\%$ (in magnitude) wavelet coefficients form other channels.
(a) nonzeroDC graphBior (SNR $15.50$ dB) (b)  zeroDC graphBior (SNR $36.24$ dB).}
\label{fig:minnesota_outp}
\end{center}
\end{figure}

\section{Conclusions}
\label{sec:conclusion}
In this paper we have presented novel graph-wavelet filterbanks that
provide a critically sampled representation with compactly supported
basis functions. The filterbanks come in two flavors: a) nonzeroDC filterbanks, and b) zeroDC filterbanks. 
The former filterbanks are designed as polynomials of the
normalized graph Laplacian matrix, and the latter filterbanks are 
extensions of the former to provide a
zero response by the highpass operators. Preliminary results showed that the filterbanks 
are useful not only for arbitrary graph but 
also to the standard regular signal processing domains. 
Extensions of this work will focus on the application
of these filters to different scenarios, including, 
for example,
social network analysis, sensor networks etc. 

\bibliographystyle{IEEEbib}
\bibliography{refs}

\begin{thebibliography}{10}

\bibitem{Hammond'09}
David~K. Hammond, Pierre Vandergheynst, and R{\'e}mi Gribonval,
\newblock ``Wavelets on graphs via spectral graph theory,''
\newblock {\em Applied and Computational Harmonic Analysis}, vol. 30, no. 2,
  pp. 129--150, Mar 2011.

\bibitem{SunilTSP}
S.K. Narang and Ortega A.,
\newblock ``Perfect reconstruction two-channel wavelet filter-banks for graph
  structured data,''
\newblock {\em IEEE trans. on Sig. Proc.}, vol. 60, no. 6, June 2012.

\bibitem{agaskar_icassp}
A.~Agaskar and Y.~M. Lu,
\newblock ``Uncertainty principles for signals defined on graphs: {Bounds} and
  characterizations,''
\newblock in {\em ICASSP}, Kyoto, Japan, Mar. 2012, pp. 3493--3496.

\bibitem{Crovella'03}
M.~Crovella and E.~Kolaczyk,
\newblock ``Graph wavelets for spatial traffic analysis,''
\newblock in {\em INFOCOM 2003}, Mar 2003, vol.~3, pp. 1848--1857.

\bibitem{Coifman'06}
R.~Coifman and M.~Maggioni,
\newblock ``{Diffusion wavelets},''
\newblock {\em Applied and Computational Harmonic Analysis}, vol. 21, pp.
  53--94, 2006.

\bibitem{Maggioni_biorthogonal}
M~Maggioni, J.~C. Bremer, R.~R. Coifman, and A.~D. Szlam,
\newblock ``Biorthogonal diffusion wavelets for multiscale representations on
  manifolds and graphs,''
\newblock in {\em Proc. {SPIE} Wavelet XI}, Sep. 2005, vol. 5914.

\bibitem{bremer_packets}
J.~C. Bremer, R.~R. Coifman, M.~Maggioni, and A.~D. Szlam,
\newblock ``Diffusion wavelet packets,''
\newblock {\em Appl. Comput. Harmon. Anal.}, vol. 21, no. 1, pp. 95--112, 2006.

\bibitem{szlam}
A.~D. {Szlam}, M.~{Maggioni}, R.~R. {Coifman}, and J.~C. {Bremer}, Jr.,
\newblock ``{Diffusion-driven multiscale analysis on manifolds and graphs:
  top-down and bottom-up constructions},''
\newblock in {\em Proc. {SPIE} Wavelets}, Aug. 2005, vol. 5914, pp. 445--455.

\bibitem{Ramchandran'06}
W.~Wang and K.~Ramchandran,
\newblock ``Random multiresolution representations for arbitrary sensor network
  graphs,''
\newblock in {\em ICASSP}, May 2006, vol.~4, pp. IV--IV.

\bibitem{GodwinJ}
G.~Shen and A.~Ortega,
\newblock ``Transform-based distributed data gathering,''
\newblock {\em Sig. Proc., IEEE Trans. on}, vol. 58, no. 7, pp. 3802 --3815,
  july 2010.

\bibitem{Silverman}
M.~Jansen, G.~P. Nason, and B.~W. Silverman,
\newblock ``Multiscale methods for data on graphs and irregular
  multidimensional situations,''
\newblock {\em Journal of the Royal Statistical Society}, vol. 71, no. 1, pp.
  97–125, 2009.

\bibitem{narang_lifting_graphs}
S.~K. Narang and A.~Ortega,
\newblock ``Lifting based wavelet transforms on graphs,''
\newblock {\em (APSIPA ASC' 09)}, Oct. 2009.

\bibitem{gavish}
M.~Gavish, B.~Nadler, and R.~R. Coifman,
\newblock ``Multiscale wavelets on trees, graphs and high dimensional data:
  {T}heory and applications to semi supervised learning,''
\newblock in {\em Proc. Int. Conf. Mach. Learn.}, Haifa, Israel, Jun. 2010, pp.
  367--374.

\bibitem{CDF9}
A.~Cohen, I.~Daubechies, and J.-C. Feauveau,
\newblock ``Biorthogonal bases of compactly supported wavelets,''
\newblock {\em Communications on Pure and Applied Mathematics}, vol. 45, no. 5,
  pp. 485--560, 1992.

\bibitem{eigenvaluespacings}
D.~Jakobson, S.~D. Miller, I.~Rivin, and Z.~Rudnick,
\newblock ``Eigenvalue spacings for regular graphs,''
\newblock in {\em IN IMA VOL. MATH. APPL}. 1999, pp. 317--327, Springer.

\bibitem{Vetterli_book}
M.~Vetterli and J.~Kova\v{c}evic,
\newblock {\em Wavelets and subband coding},
\newblock Prentice-Hall, Inc., NJ, USA, 1995.

\bibitem{JPEG2000Compensate}
M.~D. Adams and R.~Ward,
\newblock ``Wavelet transforms in the {JPEG}-2000 standard,''
\newblock in {\em In Proc. of IEEE PacRim}, 2001, pp. 160--163.

\bibitem{Mihail'02}
M.~Mihail and C.Papadimitriou,
\newblock ``On the eigenvalue power law,''
\newblock in {\em RANDOM 2002}, Sep 2002, pp. 254--262.

\bibitem{ICASSP12Sunil}
S.K. Narang and A.~Ortega,
\newblock ``Multi-dimensional separable critically sampled wavelet filterbanks
  on arbitrary graphs,''
\newblock in {\em in ICASSP'12}, Mar 2012.

\bibitem{ron}
D.~Ron, I.~Safro, and A.~Brandt,
\newblock ``Relaxation-based coarsening and multiscale graph organization,''
\newblock {\em Multiscale Model. Simul.}, vol. 9, no. 1, pp. 407--423, Sep.
  2011.

\bibitem{gp_archive}
C.~Walshaw,
\newblock ``The graph partitioning archive,''
\newblock http://staffweb.cms.gre.ac.uk/$\sim$wc06/partition/.

\bibitem{pesenson_paley}
I.~Pesenson,
\newblock ``Sampling in {Paley-Wiener} spaces on combinatorial graphs,''
\newblock {\em Trans. Amer. Math. Soc}, vol. 360, no. 10, pp. 5603--5627, 2008.

\bibitem{SSP'12}
S.~K. Narang, Y.~H. Chao, and A.~Ortega,
\newblock ``Graph-wavelet filterbanks for edge-aware image processing,''
\newblock {\em IEEE SSP Workshop}, pp. 141--144, Aug. 2012.

\bibitem{shen2010edge}
G.~Shen, W.S. Kim, S.K. Narang, A.~Ortega, J.~Lee, and H.C. Wey,
\newblock ``Edge-adaptive transforms for efficient depth map coding,''
\newblock in {\em Picture Coding Symposium (PCS), 2010}, Dec 2010.

\bibitem{wooshik}
W.S. Kim, S.K. Narang, and A.~Ortega,
\newblock ``Graph based transforms for depth video coding,''
\newblock in {\em in ICASSP'12}, Mar 2012.

\bibitem{SSIM}
H.~R.~Sheikh Z.~Wang, A. C.~Bovik and E.~P. Simoncelli,
\newblock ``Image quality assessment: From error visibility to structural
  similarity,''
\newblock {\em IEEE Trans. on Image Proc.}, vol. 13, no. 4, 2004.

\end{thebibliography}
\end{document}